%% file: full_paper_dd_2.tex
\documentclass[letterpaper, oneside, 12pt]{article}
\usepackage{amsmath}
\usepackage{graphicx}%
\usepackage{amsfonts}%
\usepackage{amssymb}
\usepackage{subfigure}
\usepackage{rotating}
\usepackage{epstopdf}
\usepackage{url}
\usepackage{comment}
\usepackage{enumerate}
\usepackage{mathtools}
\usepackage{multirow}
\usepackage{natbib}
\usepackage[usenames,dvipsnames]{color}

\usepackage{setspace}
\allowdisplaybreaks

\usepackage[hmargin=1.0in,vmargin=1.0in]{geometry}

\usepackage{paralist}
\usepackage{colortbl}

\input{defs}

\newcolumntype{L}[1]{>{\raggedright\let\newline\\\arraybackslash\hspace{0pt}}m{#1}}
\newcolumntype{C}[1]{>{\centering\let\newline\\\arraybackslash\hspace{0pt}}m{#1}}
\newcolumntype{R}[1]{>{\raggedleft\let\newline\\\arraybackslash\hspace{0pt}}m{#1}}

    \makeatletter
    \let\@fnsymbol\@arabic
    \makeatother
\title{Bayesian Tensor Regression}
\author{Rajarshi Guhaniyogi$^{\bf*}$\thanks{Rajarshi Guhaniyogi,  Assistant Professor, Department of Applied Math \& Stat, SOE2, UC Santa Cruz, 1156 High Street, Santa Cruz, CA 95064 (E-mail: rguhaniy@ucsc.edu).} \and Shaan Qamar$^{\bf*}$\thanks{Shaan Qamar, Ph.D., (E-mail: siqamar@gmail.com).} \and David B. Dunson\thanks{David B. Dunson, Arts \& Sciences Distinguished Professor, Department of Statistical Science, 218 Old Chemistry Building, Box 90251, Duke University, Durham, NC 27708-0251 (E-mail: dunson@duke.edu)}}

\begin{document}
\maketitle
\begin{abstract}

This article proposes a Bayesian approach to regression with a scalar response against vector and tensor covariates. Tensor covariates are commonly vectorized prior to analysis, failing to exploit the structure of the tensor, and resulting in poor estimation and predictive performance.  We develop a novel class of multiway shrinkage priors for the coefficients in tensor regression models.  Properties are described, including posterior consistency under mild conditions, and an efficient Markov chain Monte Carlo algorithm is developed for posterior computation.  Simulation studies illustrate substantial gains over vectorizing or using existing tensor regression methods in terms of estimation and parameter inference.  The approach is further illustrated in a neuroimaging application.
\renewcommand\thefootnote{}\footnote{$*$\bf These authors contributed equally}
\end{abstract}

\section{Introduction}\label{sec:intro}

In many application areas, it is common to collect predictors that are structured as a multiway array or tensor.  For example, the elements of this tensor may correspond to voxels in a brain image \citep{lindquist2008statistical, lazar2008statistical, hinrichs2009spatially, ryali2010sparse}.  Existing approaches for quantifying associations between an outcome and such tensor predictors mostly fall within two groups. The first approach assesses the association between each voxel and the response independently, providing a p-value `map' \citep{lazar2008statistical}.  The p-values can be adjusted for multiple comparisons to identify `significant' sub-regions of the tensor.  Although this approach is widely used and appealing in its simplicity, clearly such independent screening approaches have key disadvantages relative to methods that take into account the joint impact of the overall tensor simultaneously.  Unfortunately, the literature on simultaneous analysis approaches is sparse.

One naive approach is to simply vectorize the tensor and then use existing methods for high-dimensional regression.  Such vectorization fails to preserve spatial structure, making it more difficult to learn low-dimensional relationships with the response.  Efficient learning is of critical importance, as the sample size is typically massively smaller than the total number of voxels.  Alternative approaches within the regression framework include functional regression and two stage approaches.  The former views the tensor as a discretization of a continuous functional predictor.  Most of the literature on functional predictors focuses on 1D functions; \cite{reiss2010functional} consider the 2D case, but substantial challenges arise in extensions to 3D due to dimensionality and collinearity among voxels.  Two stage approaches first conduct a dimension reduction step, commonly using PCA, and then fit a model using lower dimensional predictors \citep{caffo2010two}.  A clear disadvantage of such approaches is that the main principal components driving variability in the random tensor may have relatively limited impact on the response variable.  Potentially, supervised PCA could be used, but it is not clear how to implement such an approach in 3D or higher dimensions.

\cite{zhou2013tensor} propose extending generalized linear regression to include a tensor structured parameter corresponding to the measured tensor predictor. To circumvent difficulties with extensions to higher order tensor predictors, they impose additional structure on the tensor parameter, supposing it decomposes as a rank-$R$ parafac sum (see Section \ref{sec:background_framework}). This massively reduces the effective number of parameters to be estimated. They develop a penalized likelihood approach where adaptive lasso penalties may be imposed on individual margins of the parafac decomposition, focusing on good point estimation for the tensor parameter. However, their method relies heavily on cross-validated methods for selecting tuning parameters, with choices for these parameters being sensitive to the tensor dimension, the signal-to-noise ratio (degree of sparsity) and the parafac rank.

Of practical interest  is a ``self calibrating'' procedure which adapts the complexity of the model to the data.  We propose a principled method to effectively shrink unimportant voxel coefficients to zero while maintaining accuracy in estimating important voxel coefficients.  Our framework gives rise to the task of model-based rank selection, with carefully constructed shrinkage priors that naturally induce sparsity within and across ranks for optimal region selection. In addition, the need for valid measures of uncertainty on parameter (predictive) estimates is crucial, especially in settings with low or moderate sample sizes, which naturally motivates our Bayesian approach. Our approach differs from image reconstruction literature as we do not model the distribution of the tensor $\bX$ \citep{qiu2007jump}. It also differs significantly from Bayesian tensor modeling literature in which the response
 is an array/tensor \citep{dunson2009nonparametric,bhattacharya2012simplex}. 

\section{Tensor regression}

\subsection{Basic notation}
\label{sec:background_framework}
Let $\bbeta_1=(\beta_{11},\dots,\beta_{1p_1})'$ and $\bbeta_2=(\beta_{21},\dots,\beta_{2p_2})'$ be vectors of length $p_1$ and $p_2$, respectively. The vector outer product $\bbeta_1\circ\bbeta_2$ is a $p_1\times p_2$ matrix with $(i,j)$-th entry $\beta_{1i}\,\beta_{2j}$. A $D$-way outer product between vectors $\bbeta_j=(\beta_{j1},\dots,\beta_{jp_j})$, $1 \le j \le D$, is a $p_1\times\cdots\times p_D$ multi-dimensional array denoted $\bB = \bbeta_1\circ\bbeta_2\circ\cdots\circ\bbeta_D$ with entries $(\bB)_{i_1,\dots,i_D}=\prod_{j=1}^{D} \beta_{ji_j}$. Define a $\vec(\bB)$ operator as stacking  elements of this $D$-way tensor into a column vector of length $\prod_{j=1}^D p_j$. From the definition of outer products, it is easy to see
that $\vec(\bbeta_1\circ\bbeta_2\circ\cdots\circ\bbeta_D)=\bbeta_D\otimes\cdots\otimes\bbeta_1$. A $D$-way tensor $\bB\in\otimes_{j=1}^D \Re^{p_j}$ has a Tucker decomposition if it can be expressed as
\begin{align}\label{eq:tucker_decomp}
\bB=\sum_{r_1=1}^{R_1}\sum_{r_2=1}^{R_2}\cdots\sum_{r_D=1}^{R_D}\lambda_{r_1,\dots,r_D}\bbeta_1^{(r_1)}\circ\bbeta_2^{(r_2)}\circ\cdots\circ\bbeta_D^{(r_D)}
\end{align}
where $\bbeta_j^{(r_j)}$ is a $p_j$ dimensional vector, $1 \le j \le D$, and $\bLambda=(\lambda_{r_1,\dots,r_D})_{r_1,\dots,r_D=1}^{R_1,\dots,R_D}$ is referred to as the {\em core tensor}. If one considers $\{\bbeta_j^{(r_j)}; 1 \le r_j \le R_j, 1 \le j \le D\}$ as ``factor loadings'' and $\lambda_{r_1,\dots,r_D}$ to be the corresponding coefficients, then the Tucker decomposition may be thought of as a multiway analogue to factor modeling.

A rank-$R$ parafac decomposition emerges as a special case of Tucker decomposition \eqref{eq:tucker_decomp} when $R_1=R_2=\cdots=R_D=R$ and $\lambda_{r_1,\dots,r_D} = I(r_1=r_2=\cdots=r_D)$.
In particular, $\bB\in\otimes_{j=1}^D \Re^{p_j}$ assumes a rank-$R$ parafac decomposition if
\begin{align} \label{eq:parafac_decom}
\bB=\sum_{r=1}^{R}\bbeta_1^{(r)}\circ\cdots\circ\bbeta_D^{(r)}
\end{align}
where $\bbeta_j^{(r)}$ is a $p_j$ dimensional column vector as before, for $1 \le j \le D$ and $1 \le r \le R$.  These vectors are often referred to as `margins.'
The parafac decomposition is more widely used due to its relative simplicity.

\subsection{Model framework}
Let $y \in \mathcal{Y}$ denotes a response variable, with $\bz \in \mathcal{X} \subset \Re^p$ and $\bX \in \otimes_{j=1}^D \Re^{p_j}$ scalar and tensor predictors, respectively. We assume response $y$ follows an exponential family distribution
\begin{align}\label{eq:glmtens}
f(y | \btheta,\tau) = \exp\left(\frac{y\btheta -b(\btheta)}{a(\tau)} + c(y, \tau)\right)
\end{align}
with natural parameter $\btheta$, dispersion $\tau > 0$ and known functions $a(\tau)$,  $b(\btheta)$ and $c(y, \tau)$.  Usual GLMs focus on vector predictors $\bz$ and let
 $g(\Exp(y | \bz)) = \alpha + \bz'\bgamma$, for a strictly increasing canonical link function $g(\cdot)$ and model parameters $\alpha \in \Re$, $\bgamma\in\Re^p$.  To generalize this framework to also include tensor predictor $\bX$, we let
 \begin{align}\label{mean:par1}\
g(\Exp(y|\bz,\bX)) =\alpha+\bz'\bgamma+\langle \bX,\bB\rangle, \quad
\langle \bX,\bB\rangle = \vec(\bX)'\vec(\bB)
\end{align}
where $\bB \in \otimes_{j=1}^D \Re^{p_j}$ is the tensor parameter corresponding to measured tensor predictor $\bX$.  For concreteness, we focus on linear regression with $g$ the identity link.


The coefficient tensor $\bB$ has $\prod_{j=1}^D p_j$ elements, necessitating substantial dimensionality reduction.
A rank-1 parafac decomposition assumes $\bB=\bbeta_1\circ\cdots\circ\bbeta_D$ and $\vec(\bB)=\bbeta_D\otimes\cdots\otimes\bbeta_1$. This reduces to modeling $g(\Exp(y|\bz,\bX))=\alpha+\bz'\bgamma+\bbeta_1'\bX\bbeta_2$ when $D = 2$, corresponding to the bilinear model considered in \cite{hung2013matrix}. Since only the single parameter vector $\bbeta_j$ captures signal along the $j$th dimension, a rank-1 assumption severely limits flexibility ruling out interactions among dimensions.

Following \cite{zhou2013tensor}, we use a more flexible rank-$R$ parafac decomposition for $\bB = \sum_{r=1}^R \bbeta_1^{(r)}\circ\cdots\circ\bbeta_D^{(r)}$ introduced in \eqref{eq:parafac_decom} with $\bbeta_j^{(r)} \in \Re^{p_j}$, $1 \le j \le D$, and $1\le r \le R$. Expression \eqref{mean:par1} then becomes
\begin{align}\label{newtensor}
\begin{aligned}
g(\Exp(y|\bz,\bX)) &=\alpha+\bz'\bgamma+\Big\langle \bX, \sum_{r=1}^R \bbeta_1^{(r)}\circ\cdots\circ\bbeta_D^{(r)}\Big\rangle\\
&=\alpha+\bz'\bgamma + \sum_{(i_1,\dots,i_D)}
(\bX)_{i_1,\dots,i_D} (\bB)_{i_1, \dots, i_D}
\end{aligned}
\end{align}
where voxel $(\bX)_{i_1,\dots,i_D}$ of the tensor predictor has corresponding parameter
\begin{equation}\label{eq:voxel_coefficient}
(\bB)_{i_1, \dots, i_D} = \sum_{r=1}^{R} \prod_{j = 1}^D \beta_{j, i_j}^{(r)}, \quad (i_1, \dots, i_D) \in \mathcal{V}_\bB = \otimes_{j=1}^D \{1, \dots, p_j\}.
\end{equation}
The model is therefore nonlinear in the  parameters defining $\bB$. A hierarchical specification is completed by placing appropriate priors over unknown model parameters. Existing priors may be chosen for $\alpha$ and $\bgamma$, but specification of the prior over the tensor parameters is nontrivial; see Sections and \ref{sec:multprior} and \ref{sec:multiway_criteria}.

Under the assumed rank-$R$ decomposition for $\bB$, model \eqref{newtensor} requires estimating $p + R\sum_{j=1}^D p_j$ as opposed to $p + \prod_{j=1}^D p_j$ parameters for the unstructured vectorized (saturated) model.  One wonders whether such dramatic dimension reduction retains sufficient flexibility.  In particular, we are interested in identifying geometric sub-regions of the tensor across which the coefficients are not close to zero, with the remaining elements being very close to zero.  We would also like to accurately estimate the coefficients in these sub-regions.  We have observed good performance in addressing these goals in extensive simulation studies summarized Section \ref{sec:simulation_studies}, consistent with our theoretical analyses in Section \ref{sec:tensor_theory}.

\subsection{Model identifiability}

From model \eqref{newtensor} it is clear that only voxel-level coefficients are identified and not the individual tensor margins defining their product-sum given in \eqref{eq:voxel_coefficient}. In the tensor setting, identifiability restrictions are understood  in light of the following indeterminacies:
\begin{enumerate}
\item \emph{Scale indeterminacy}: for each $r=1,\dots,R$, define $\blambda_r = (\lambda_{1r}, \dots, \lambda_{Dr})$ such that $\prod_{j=1}^{D}\lambda_{jr}=1$. Then replacing $\bbeta_{j}^{(r)}$ by $\lambda_{jr} \bbeta_{j}^{(r)}$ leaves the tensor parameter $\bB$ unaltered.
\item \emph{Permutation  indeterminacy}: $\sum_{r=1}^{R}\circ_{j=1}^{D}\bbeta_{j}^{(r)}=\sum_{r=1}^{R}\circ_{j=1}^{D}\bbeta_{j}^{(P(r))}$
    for any permutation $P(\cdot)$ of $\{1,2,\dots,R\}$. In particular, this implies that
$\circ_{j=1}^{D}\bbeta_{j}^{(r)}$ are not identifiable for $r=1,\dots,R$.
\item \emph{Orthogonal transformation indeterminacy} ($D = 2$ only): for any orthonormal matrix $\bO$, one has $(\bbeta_{1}^{(r)}\bO)\circ(\bbeta_{2}^{(r)}\bO)=\bbeta_{1}^{(r)}\otimes \bbeta_{2}^{(r)}$.
\end{enumerate}
For $D>2$, imposing the following $(D-1)R$ constraints ensures identifiability of the margin parameters comprising the rank-R parafac decomposition:
\begin{align}\label{eq:identifiability_constraints}
&\beta_{j,1}^{(r)}=1, ~1 \le j < D, ~1\le r \le R, \qquad \text{and} \qquad \beta_{D,1}^{(1)}>\cdots>\beta_{D,1}^{(R)}.
\end{align}
For our proposed Bayesian method, we seek accurate estimation and inferences on $\bB$ along with state-of-the-art predictive performance. Neither of these goals rely on identifiability of the tensor margins, $\bbeta_{j}^{(r)}$, and hence we avoid identifiability restrictions on these parameters.  The lack of restrictions simplifies the design of efficient computational algorithms.

\section{Multiway shrinkage priors}\label{sec:multprior}

\subsection{Vector shrinkage priors}
There has been recent interest in high-dimensional regression with vector predictors, choosing priors which shrink small coefficients towards zero while minimizing shrinkage of large coefficients.  Many of these priors can be expressed as a global-local (GL) scale mixtures \citep{polson2012local} with
\begin{align} \label{eq:local_global_prior}
\theta_j \sim \mathrm{N}(0,\psi_j\tau),
\quad \psi_j \sim g,
\quad \tau \sim h,
\end{align}
where $(\theta_1,\ldots,\theta_p)$ is a coefficient vector, $\tau$ is a global scale and $\psi_j$ is a local-scale.  When $g$ is a mixture of two components, with one concentrated near zero and the other away from zero, a spike and slab prior is obtained.  Many other choices of $g$ and $h$ have been considered.  Although the GL family is widely used and versatile, \cite{bhattacharya2014dirichlet} note advantages in drawing the local scales jointly.  In particular, they propose to let
\begin{align*}
\theta_j \sim \mathrm{DE}(\cdot | \phi_j \tau),
\quad \phi_j \sim \mathrm{Dirichlet}(a,\ldots,a),
\quad \tau \sim h.
\end{align*}
where $\mathrm{DE}(\cdot)$ denote the double-exponential distribution.  For small $a$ and large $p$, the Dirichlet($a,\ldots,a$) prior has the property of favoring many values close to zero with a few much larger values, but with $\sum_j \phi_j = 1$.

\subsection{Multiway priors} \label{sec:multiway_criteria}
We propose a new class of multiway shrinkage priors in the generalized linear model setting with tensor valued predictors. Assuming tensor parameter $\bB$ admits a rank-$R$ parafac decomposition, model \eqref{newtensor} results in
voxel-level coefficients that are a nonlinear function of the corresponding tensor margin parameters (see \eqref{eq:voxel_coefficient}). Moreover, this implies simultaneous shrinkage on each of the $\prod_{j=1}^D p_j$ voxel coefficients as imposed by the prior over $R\sum_{j=1}^D p_j$ parameters. This necessitates careful prior specification on the tensor margins $\{\bbeta_j^{(r)}; 1 \le j \le D, 1 \le r \le R\}$ such that the induced voxel-level prior has adequate tails so as to prevent over shrinkage.

There are a number of desirable characteristics for a multiway prior on the tensor margins in the absence of prior information that certain elements of the tensor are more likely to be important.  In particular, it is important to ensure that
\begin{enumerate}
\item For each $r = 1, \dots, R$, $\big(\beta_{1, i_1}^{(r)},\dots, \beta_{D, i_D}^{(r)}\big)$ and $\big(\beta_{1, k_1}^{(r)},\dots,\beta_{D, k_D}^{(r)}\big)$ are equal in distribution, for any $(i_1,\dots,i_D), (k_1,\dots,k_D) \in \mathcal{V}_\bB \times \mathcal{V}_\bB$ and $(i_1, \dots, i_D) \ne (k_1, \dots, k_D)$.
\item Shrinkage towards a low rank decomposition, with the model adapting to the complexity and signal in the data, effectively deleting unnecessary dimensions.
\item The prior should favor recovery of contiguous geometric subregions of the tensor across which the voxel observations are predictive of the response.
\end{enumerate}
In addition, the proposed multiway shrinkage prior must have a structure that facilitates efficient and reliable model fitting.

\subsection{The multiway Dirichlet GDP prior}
\label{sec:mdgdp_prior}

There are many ways of specifying priors over tensor margins $\bbeta_j^{(r)}$ to satisfy the listed criteria. In this article we propose a particular choice which we deem the multiway Dirichlet generalized double Pareto (M-DGDP) prior. The M-DGDP prior induces shrinkage across components in an exchangeable way,
setting $\tau_r = \phi_r\tau$  as the global scale for component $r = 1, \dots, R$, with $\tau \sim \mathrm{Ga}(a_\tau, b_\tau)$ and $\Phi = (\phi_1, \dots, \phi_R) \sim \mathrm{Dirichlet}(\alpha_1, \dots, \alpha_R)$.
In addition, define $\bW_{jr} = \mathrm{diag}(w_{jr, 1}, \dots, w_{jr,p_j})$ for $1 \le j \le D$ and $1 \le r \le R$ as local (margin, component-specific) scale parameters. The hierarchical margin-level prior is  given by
\begin{align}\label{eq:M1}
\bbeta_j^{(r)} \sim \mathrm{N}\big(0, (\phi_r \tau) \bW_{jr}\big),\quad w_{jr,k} \sim \mathrm{Exp}(\lambda_{jr}^2 / 2),\quad
\lambda_{jr} \sim \mathrm{Ga}(a_\lambda, b_\lambda).
\end{align}

Additional flexibility in estimating $\bB_r = \{ \bbeta_j^{(r)}; 1 \le j \le D\}$ is  accommodated by modeling heterogeneity within margins via element-specific scaling $w_{jr,k}$. A common rate parameter $\lambda_{jr}$ encourages sharing of information between the margin elements. Collapsing over the element-specific scales, $\beta_{j,k}^{(r)} | \lambda_{jr}, \phi_r, \tau \overset{\rm iid}{\sim} \mathrm{DE}(\lambda_{jr} / \sqrt{\phi_r \tau})$, $1 \le k \le p_j$. Prior \eqref{eq:M1} leads to a GDP prior \citep{armagan2013generalized} on the individual margin coefficients.

\subsection{Prior hyper-parameter elicitation} \label{sec:dirichlet_alpha}

It is important to assess how the shrinkage prior \eqref{eq:M1} on the margins impacts the induced prior on the voxel coefficients.  Unfortunately,
the distribution of the voxel-level coefficients \eqref{eq:voxel_coefficient} is not available in closed form. However, the voxel-level variance under the M-DGDP prior \eqref{eq:M1} is given by
\begin{align*}
\var(\bB_{i_1, \dots, i_D}) &= \Exp\bigg( \var \bigg\{ \sum_{r = 1}^R \prod_{j=1}^D \bbeta_{j, i_d}^{(r)} \,|\, \bW, \Lambda, \Phi, \tau \bigg\}\bigg) \\
& = \Exp_{\Phi} \bigg(\sum_{r=1}^R \phi_r^D \,\Exp_{\tau}\{\tau^D\} \, \Exp_{\Lambda_{\cdot, r}} \bigg\{ \Exp_{\bW_r | \bLambda_r} \bigg(\prod_{j=1}^D w_{jr,i_j}\bigg)\bigg\} \bigg) \\
& = \frac{\Gamma(\alpha_0 + D)}{\Gamma(\alpha_0) \, b_\tau^D} (2 C_\lambda)^D \, \Exp_{\Phi} \bigg(\sum_{r=1}^R \phi_r^D \bigg)
\end{align*}
where the last step follows from the MGF of a Gamma distributed random variable.  The following Lemma provides lower and upper bounds on the variance, which are useful in hyperparameter elicitation.

\begin{lemma} \label{lem:dl_gdp_var}
Under M-DGDP shrinkage prior \eqref{eq:M1} and for $D > 1$, if $\alpha_1 = \cdots = \alpha_R = c / R, \, c \in \mathbb{N}_{+}$, and  with constants $C_\lambda = b_\lambda^{2} / \big((a_\lambda -1)(a_\lambda - 2)\big)$, $a_\lambda > 2$, $A_\tau = \exp((D^2-3D)/2)$, then the voxel-level variance is bounded below by $R \alpha_1^D (2C_\lambda  / b_\tau)^D$ and above by $A_\tau (2 C_\lambda / b_\tau)^D\, \exp(\alpha_1 \, R D)$.
\end{lemma}
\noindent{\bf Proof:} See Appendix \ref{sec:prior_moment_pfs}.

Hyperparameters in the Dirichlet component of the multiway prior \eqref{eq:M1} play a key role in controlling dimensionality of the model, with smaller values favoring more component-specific scales $\tau_r \approx 0$, and hence collapsing on an effectively lower rank factorization.
Figure \ref{fig:dirichlet_concentration} plots realizations from the Dirichlet distribution when $R = 3$ for different concentration parameters\footnote{For simplicity we have assumed $\alpha_1 = \cdots = \alpha_R = \alpha$.} $\alpha$. As $\alpha \downarrow 0$, points increasingly tend to concentrate around vertices of the $\mathcal{S}^{R-1}$ probability simplex, $R > 1$, leading to increasingly sparse realizations\footnote{This notion of sparsity is made precise in \cite{yang2014minimax}.}.
\begin{figure}[h]
\centering
\begin{tabular}{ccc}
\subfigure[$\balpha = (0.2, 0.2, 0.2)$]{\includegraphics[width=0.32\linewidth]{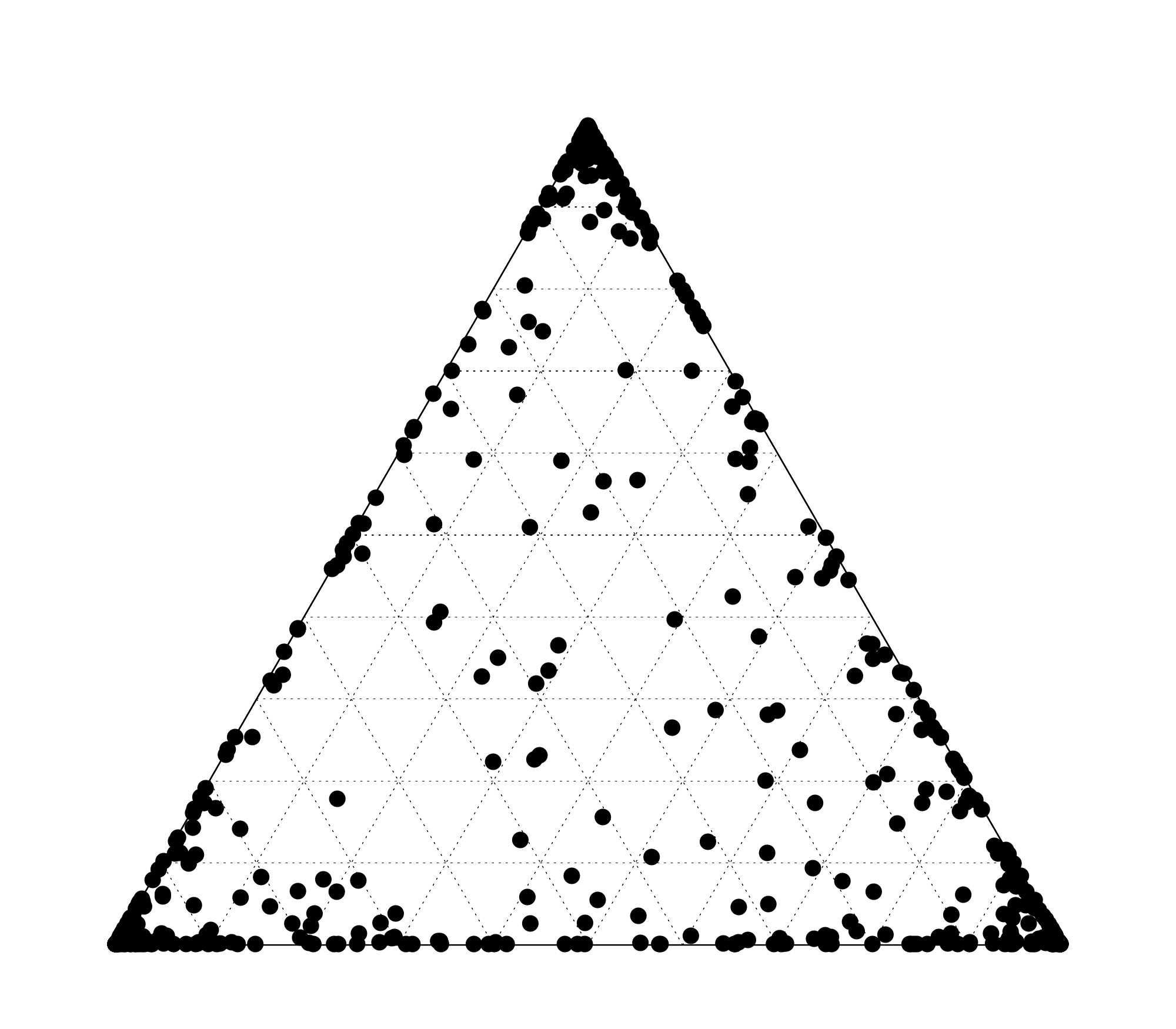}} &
\subfigure[$\balpha = (0.3, 0.3, 0.3)$]{\includegraphics[width=0.32\linewidth]{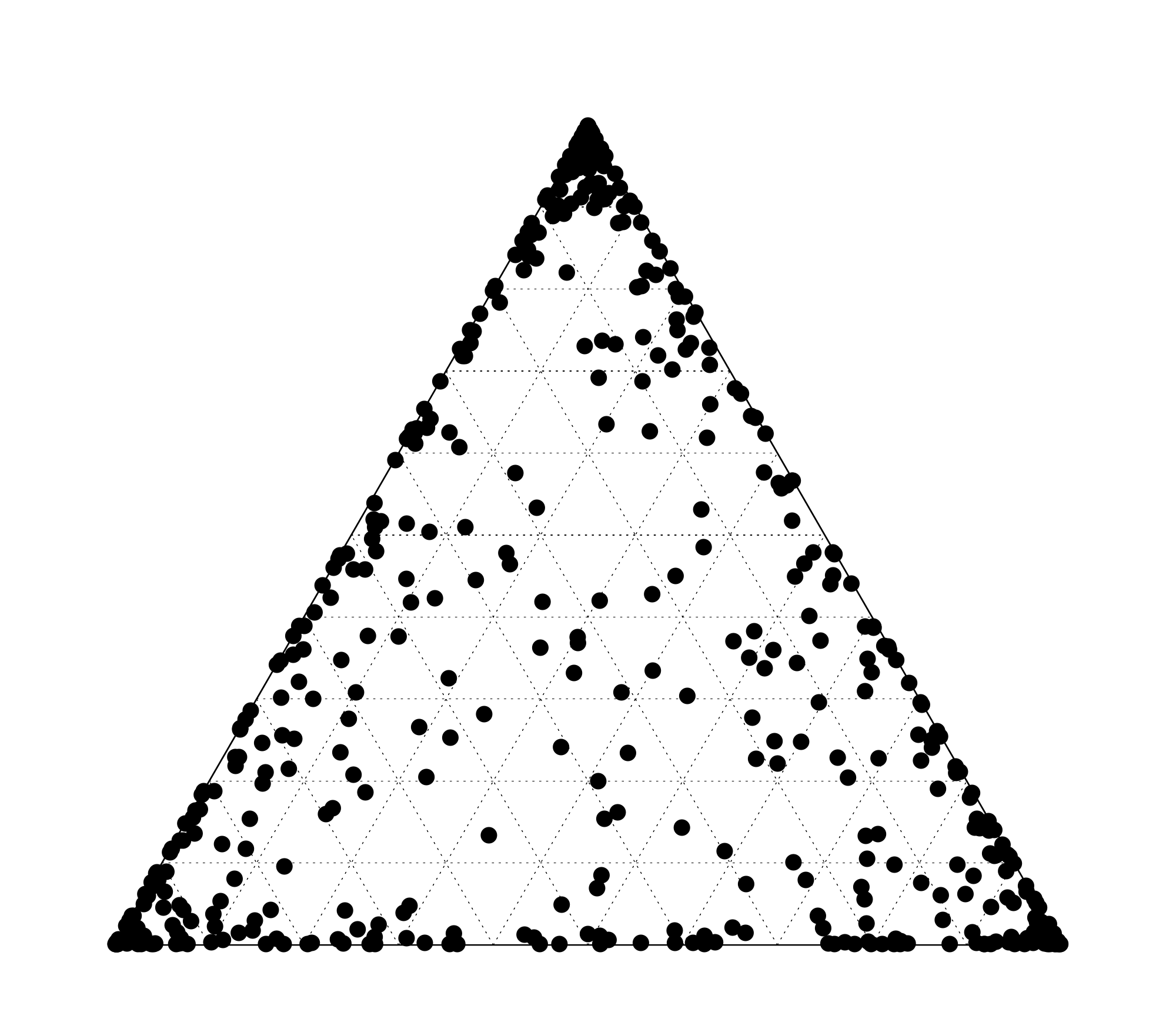}} &
\subfigure[$\balpha = (0.5, 0.5, 0.5)$]{\includegraphics[width=0.32\linewidth]{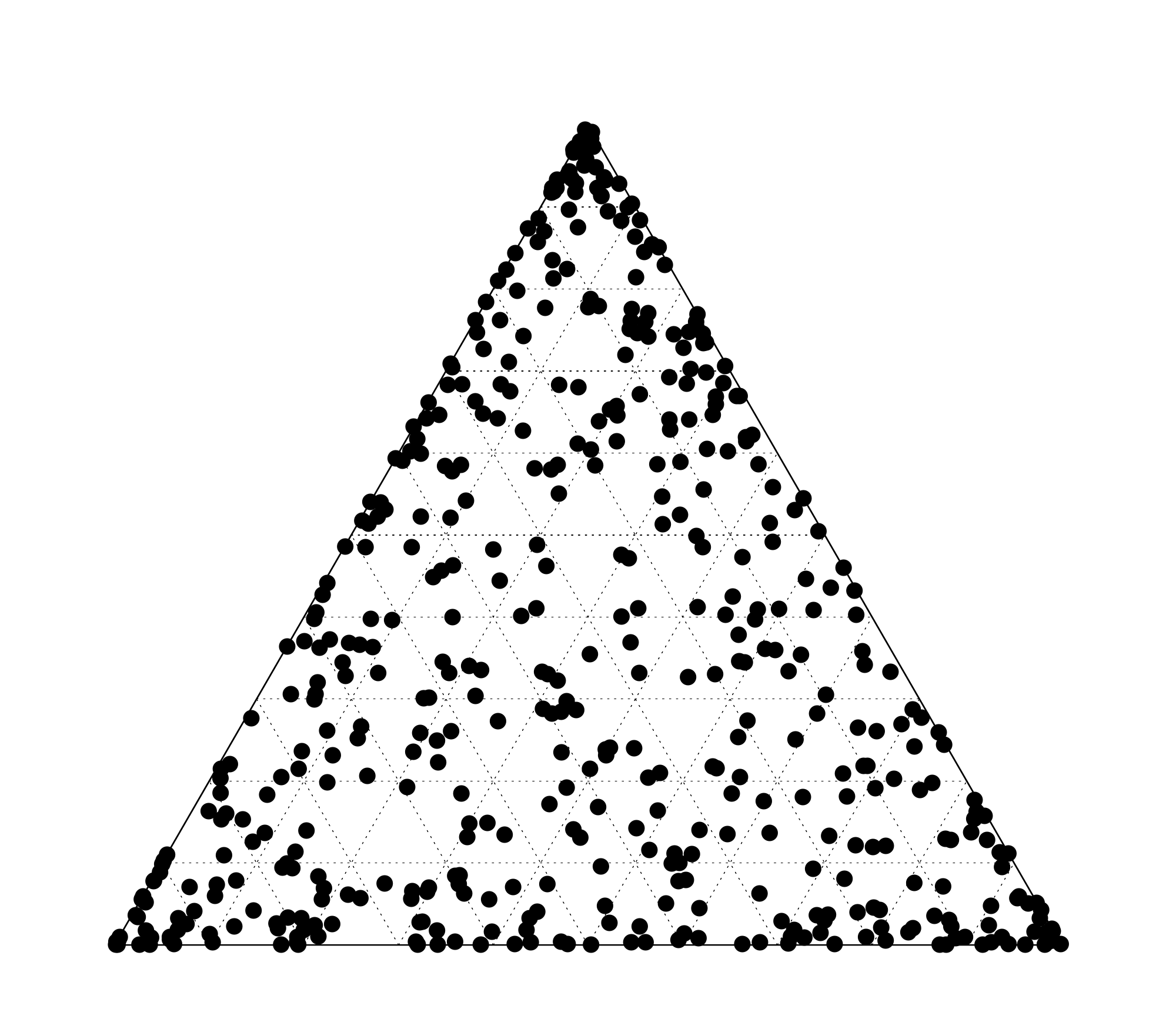}}
\end{tabular}
\caption{Visualization of points in the $\mathcal{S}^2$ probability simplex for 500 independent realizations of $X \sim \mathrm{Dirichlet}(\balpha)$.}
\label{fig:dirichlet_concentration}
\end{figure}

We allow $\alpha$ to be unknown by choosing a discrete uniform prior over a grid $\mathcal{A}$, which we choose to be 10 values equally spaced between $R^{-D}$ and $R^{-0.10}$ as a default.  \cite{armagan2013generalized} study various choices of $(a_\lambda, \zeta = b_\lambda / a_\lambda)$ that lead to desirable shrinkage properties, such as Cauchy-like tails for $\beta_{j,k}^{(r)}$ while retaining Laplace-like shrinkage near zero. Empirical results from simulation studies across a variety of settings in Section \ref{sec:simulation_studies} reveal no strong sensitivity to choices for hyper-parameters $a_\lambda, b_\lambda$. From Lemma 3.1, setting $a_\lambda = 3$ and $b_\lambda =\sqrt[2D]{a_\lambda}$ avoids overly narrow variance of the induced prior on $\bB_{i_1, \dots, i_D}$.  Table 1 and Figures 2-3 illustrate the induced prior on the tensor elements under our default choices.

\begin{table}[!h]
\centering
\begin{tabular}{l | c | ccccc }
& $R$ & 5\% & 25\% & 50\% & 75\% & 95\% \\
\hline
\multirow{3}{*}{$D = 2$}
 & 1 & 0.001 & 0.011 & 0.057 & 0.254 & 1.729 \\
 & 5 & 0.004 & 0.040 & 0.164 & 0.595 & 3.332 \\
 & 10 & 0.005 & 0.058 & 0.237 & 0.852 & 4.635 \\
\hline
\multirow{3}{*}{$D = 3$}
 & 1 & 0.000 & 0.001 & 0.010 & 0.072 & 0.917 \\
 & 5 & 0.000 & 0.009 & 0.061 & 0.341 & 3.382 \\
 & 10 & 0.001 & 0.017 & 0.111 & 0.608 & 5.996 \\
\end{tabular}
\caption{Percentiles for $|\bB_{i_1, \dots, i_D}|$ under the M-DGDP prior with default $a_\lambda = 3$, $b_\lambda = \sqrt[2D]{a_\lambda}$, $b_\tau = \alpha R^{1/D} \,(v = 1)$ and $\alpha = 1/R$. Statistics are displayed as the dimension $D$ of the tensor and its parafac rank decomposition $R$ vary. }
\label{tab:prior_quantile_stats}
\end{table}

\begin{figure}[h]
\centering
\setlength{\tabcolsep}{1pt}
\begin{tabular}{C{0.1\columnwidth}C{0.29\columnwidth}C{0.29\columnwidth}C{0.29\columnwidth}}
& $R = 1$ & $R = 5$ & $R = 10$ \\
$D = 2$ & \includegraphics[width=\linewidth]{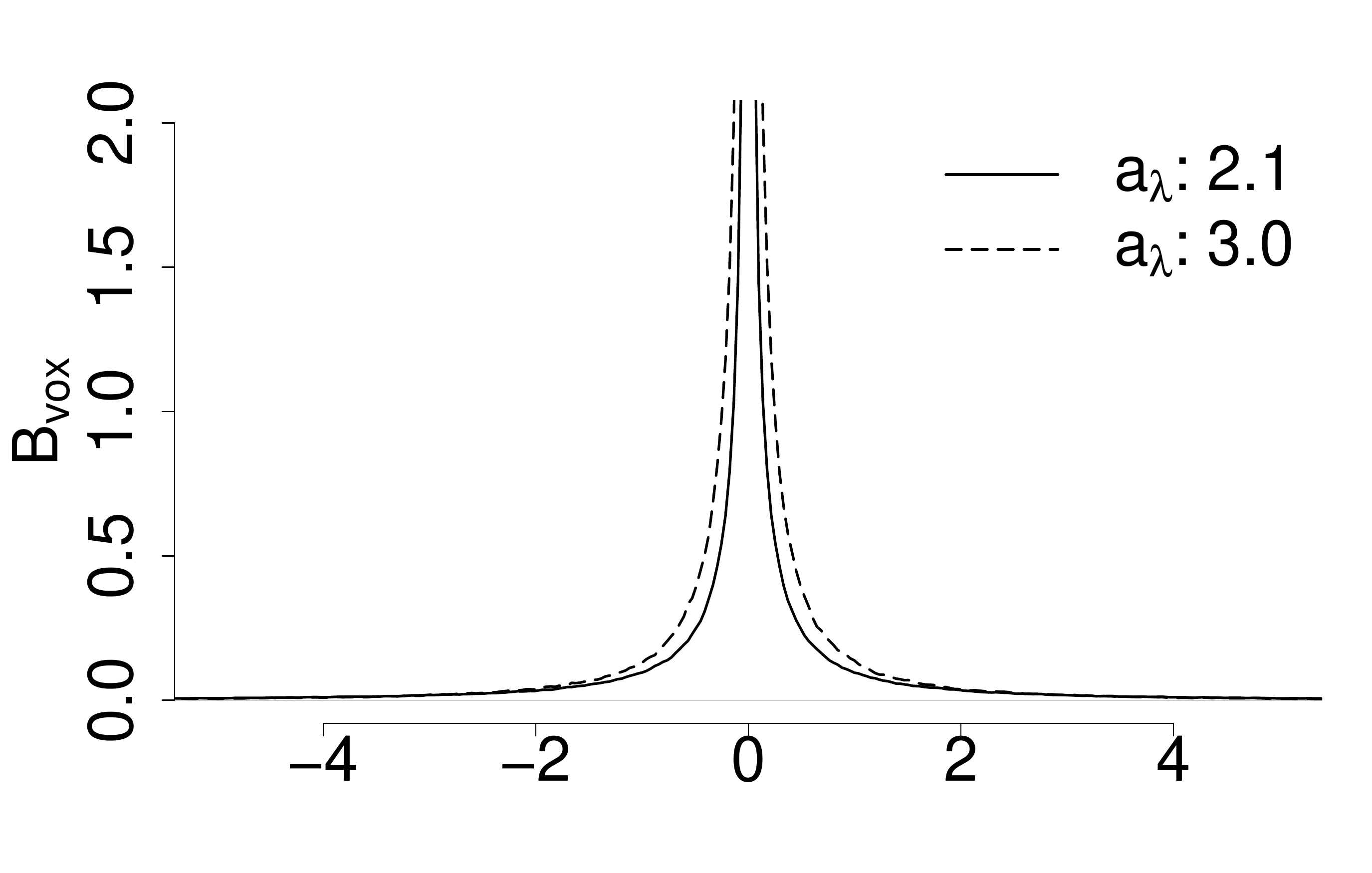} & \includegraphics[width=\linewidth]{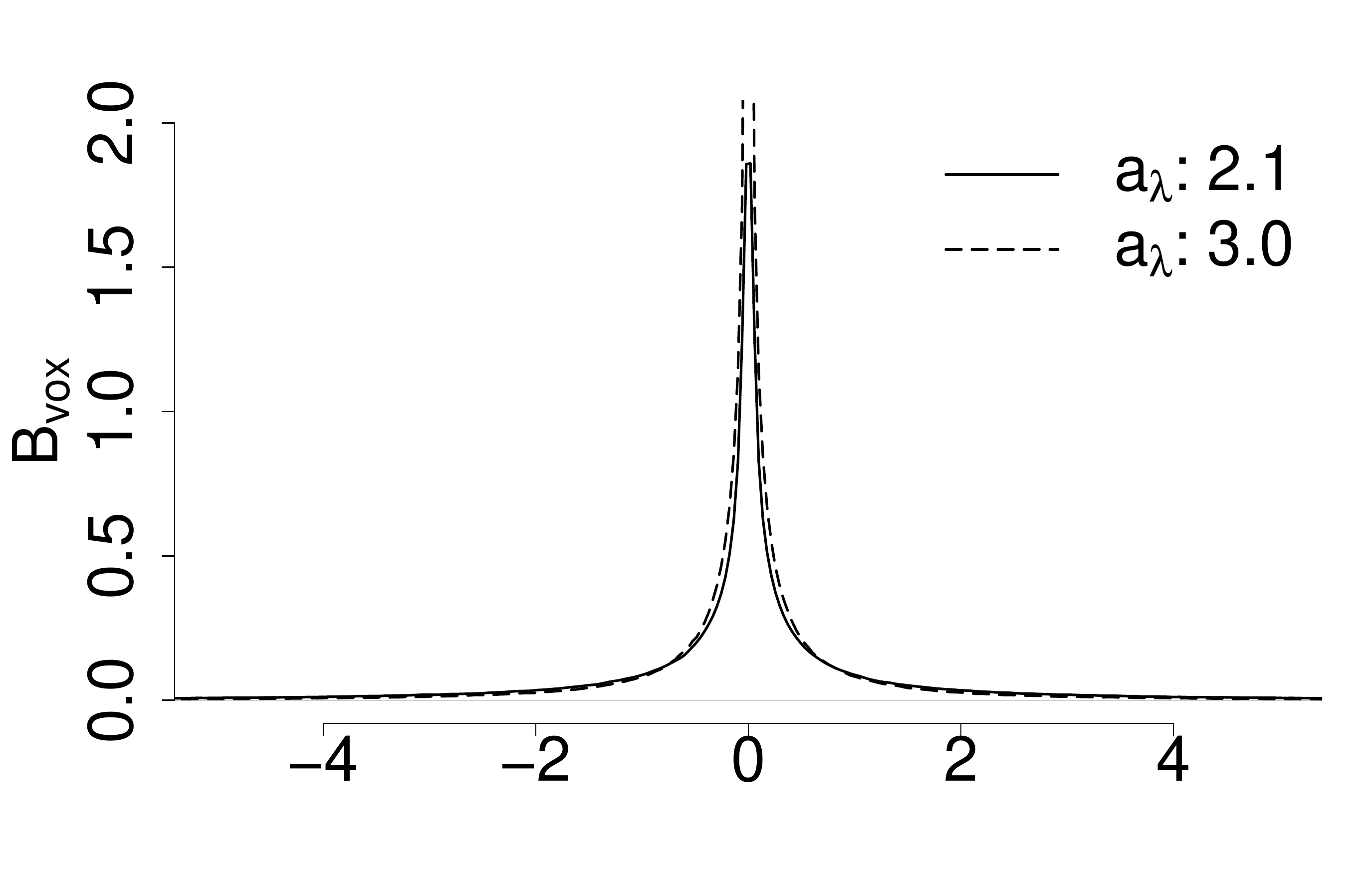} & \includegraphics[width=\linewidth]{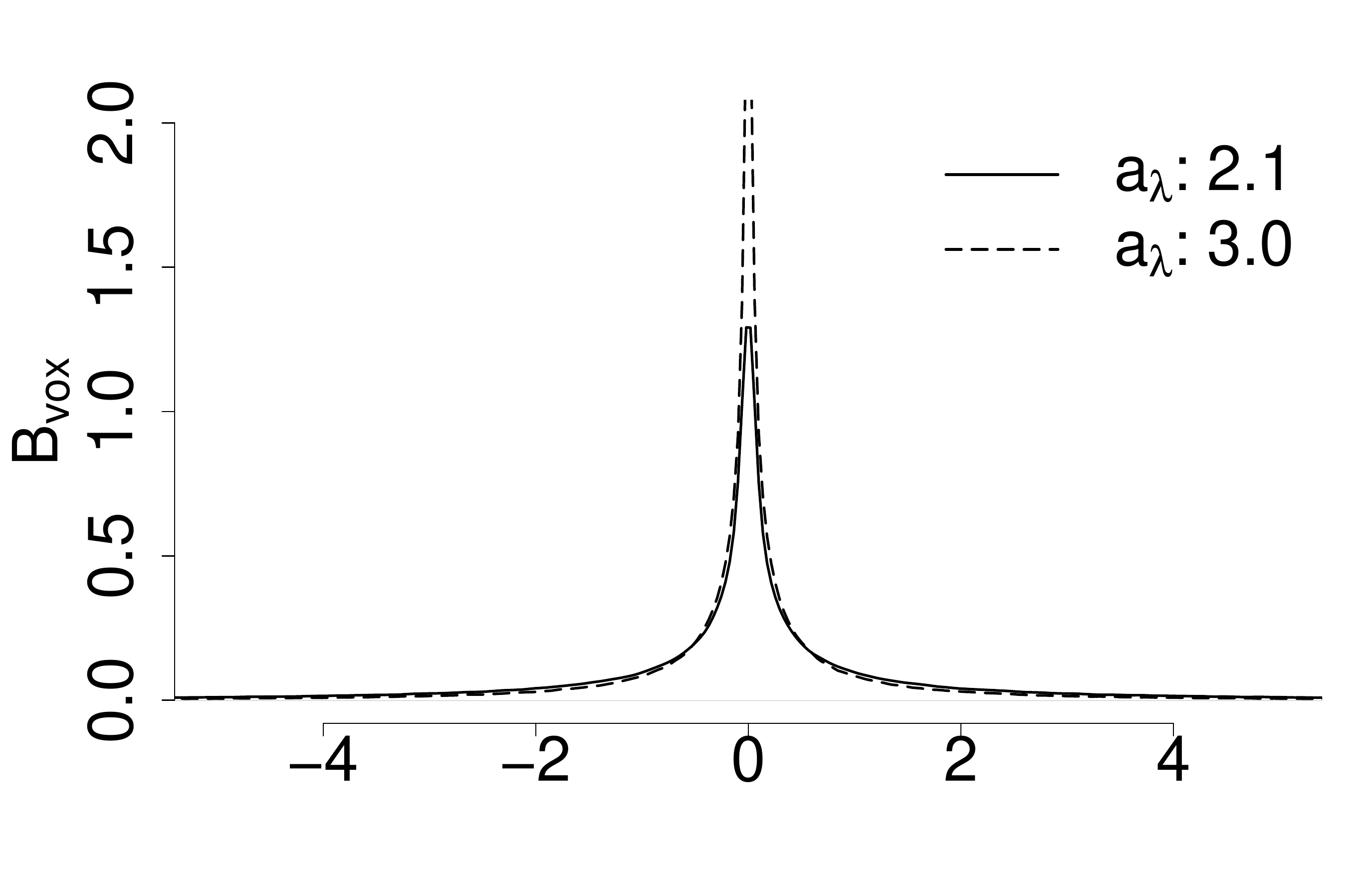} \\
$D = 3$ & \includegraphics[width=\linewidth]{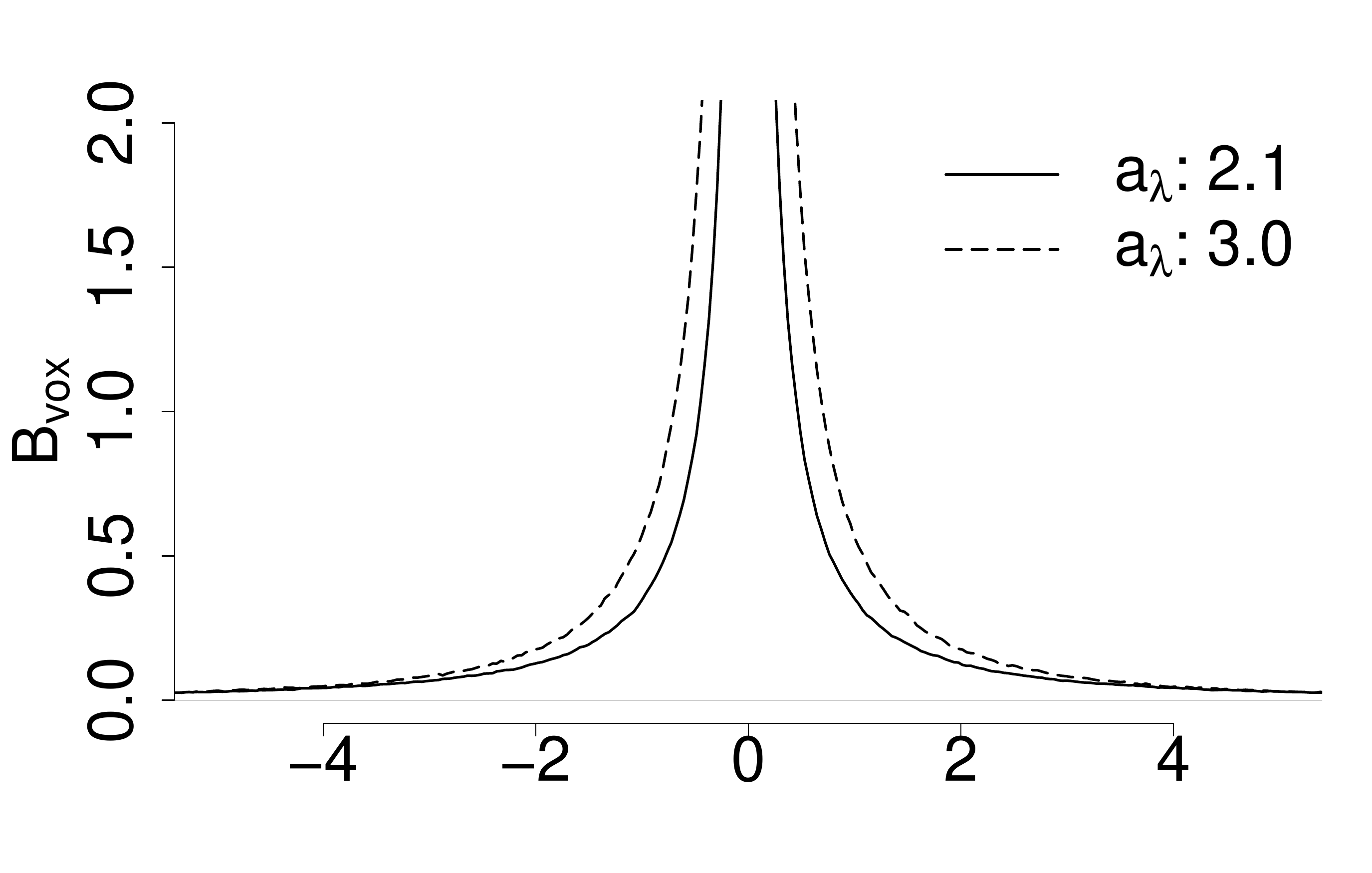} & \includegraphics[width=\linewidth]{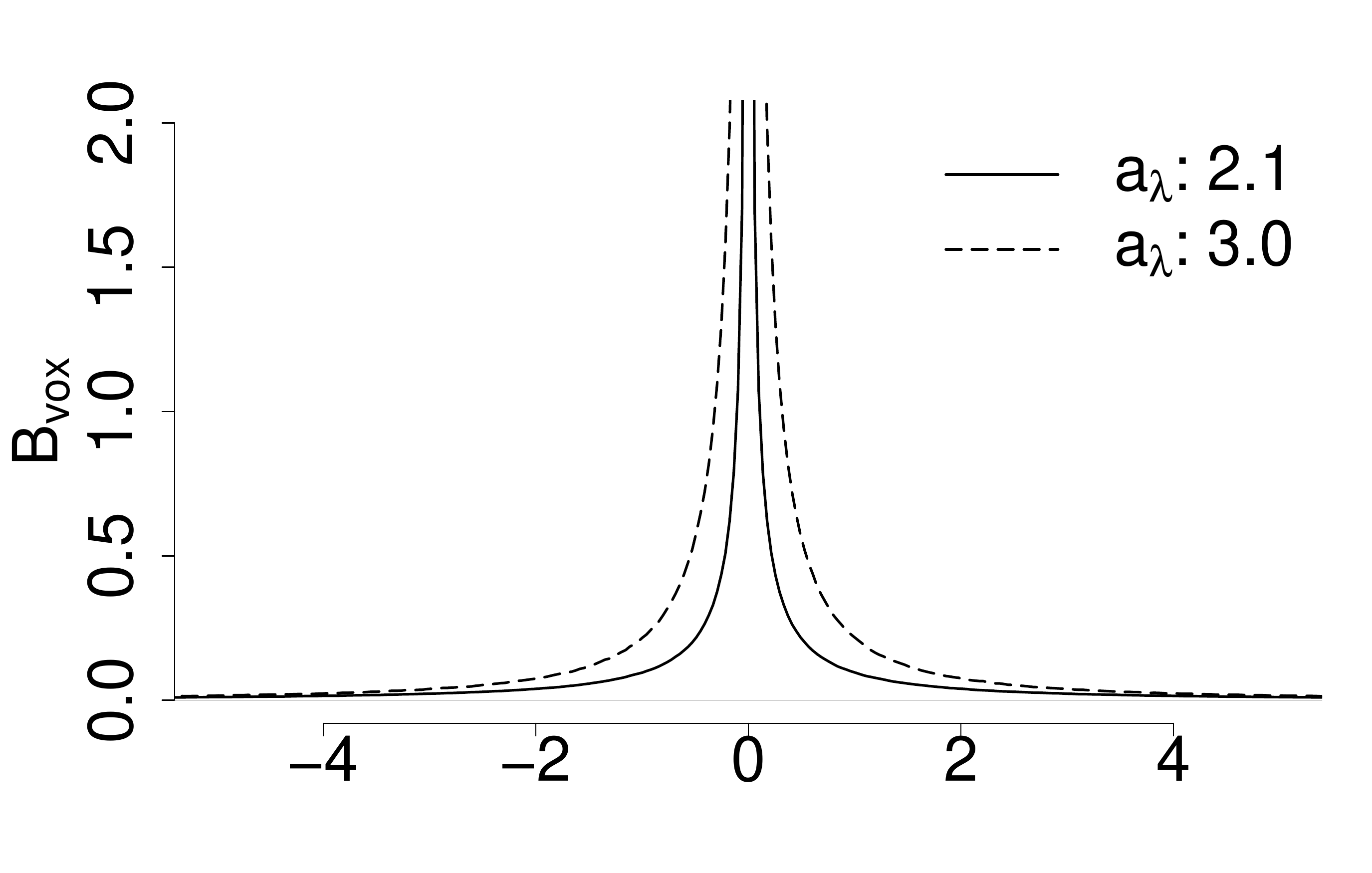} & \includegraphics[width=\linewidth]{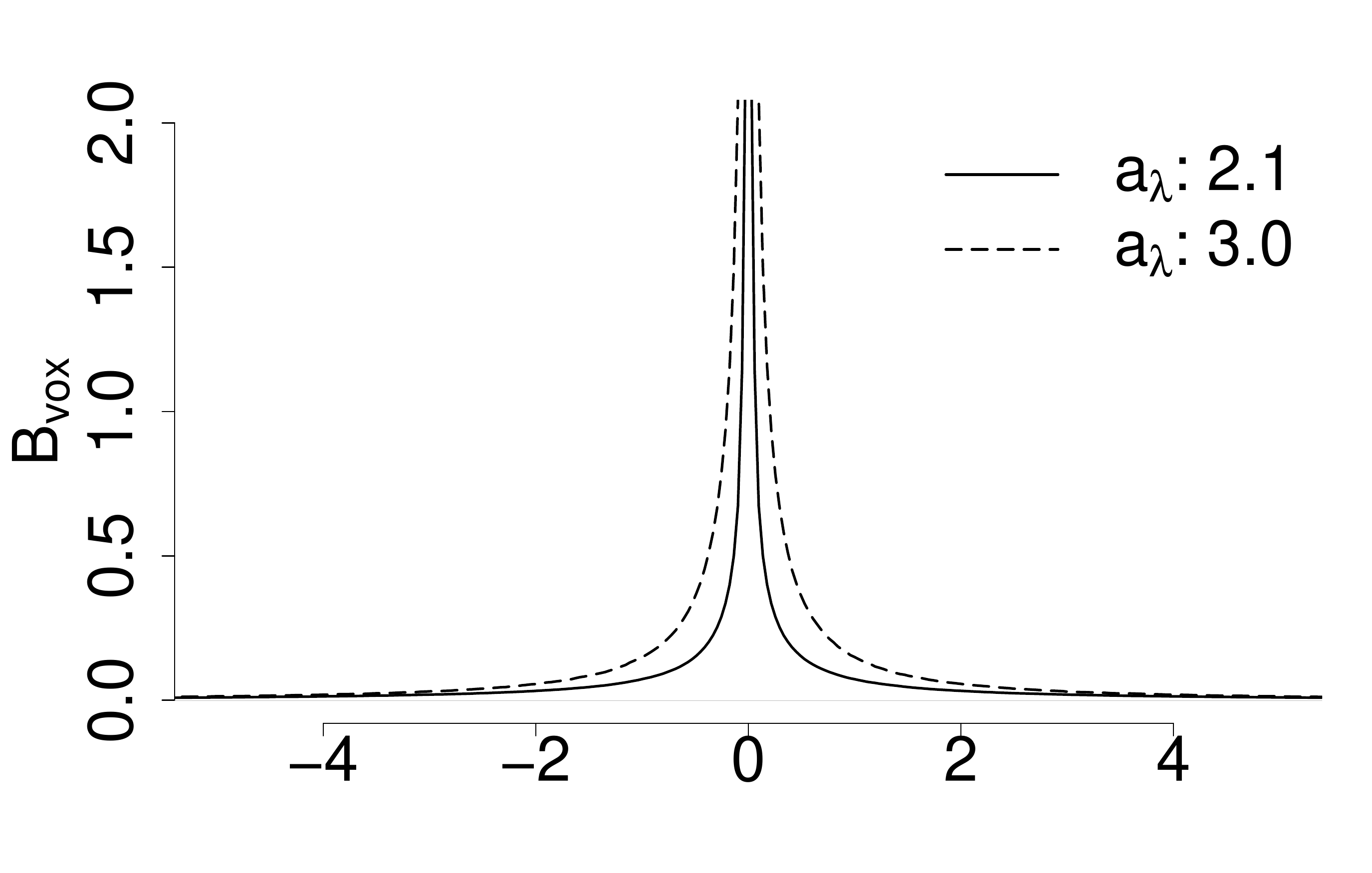}
\end{tabular}
\caption{Induced voxel level prior distribution for default specification as a function of $a_\lambda$, with $b_\lambda = \sqrt[2D]{a_\lambda}$, $b_\tau = \alpha R^{1/D}$ and $\alpha = 1/R$.}
\label{fig:simu_dl_gdp_prior}
\end{figure}

\begin{figure}[h]
\centering
\setlength{\tabcolsep}{1pt}
\begin{tabular}{C{0.32\columnwidth}C{0.32\columnwidth}C{0.32\columnwidth}}
$a_\lambda = 2$ & $a_\lambda = 3$ & $a_\lambda = 5$ \\
\includegraphics[width=\linewidth]{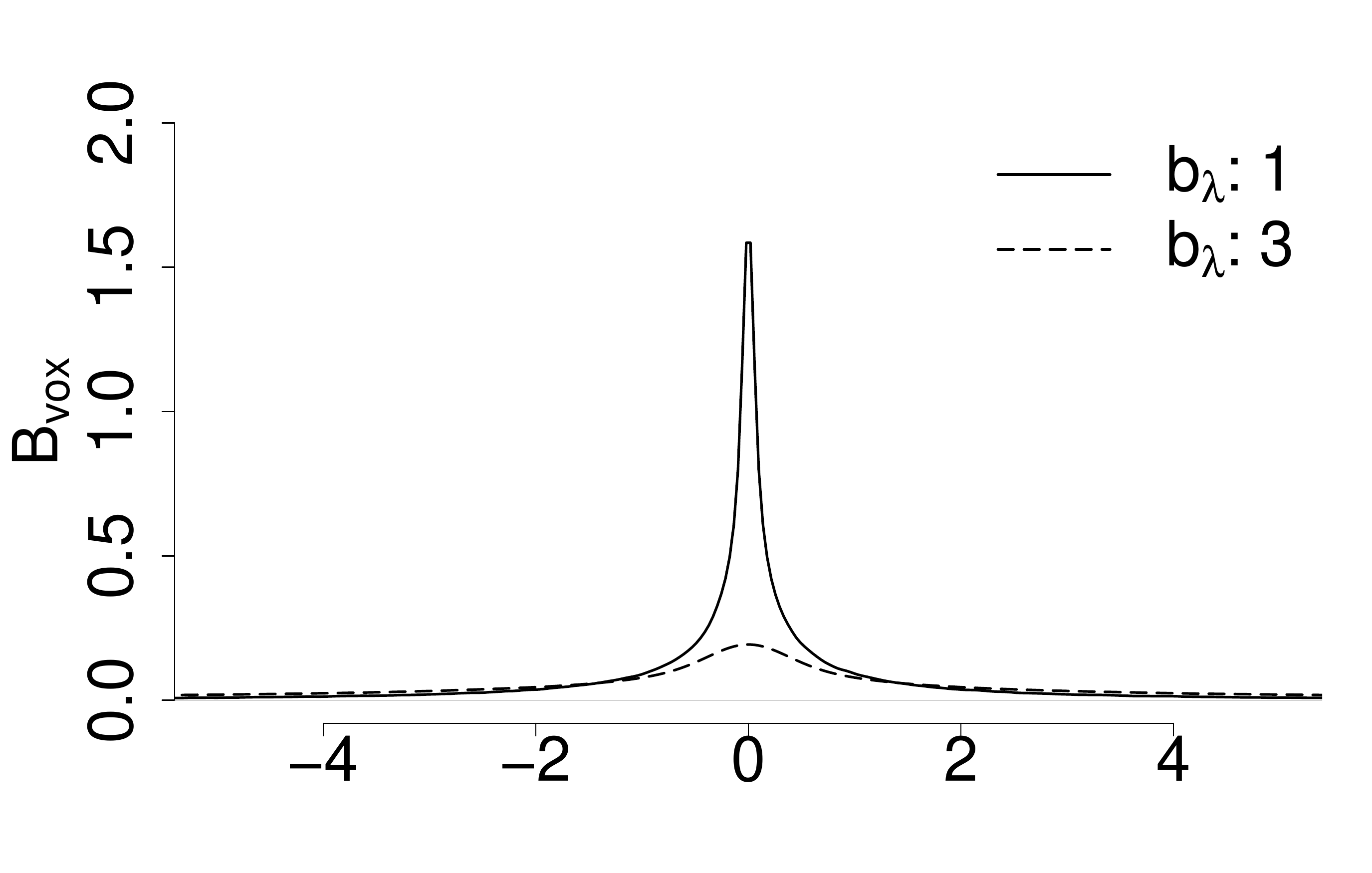} & \includegraphics[width=\linewidth]{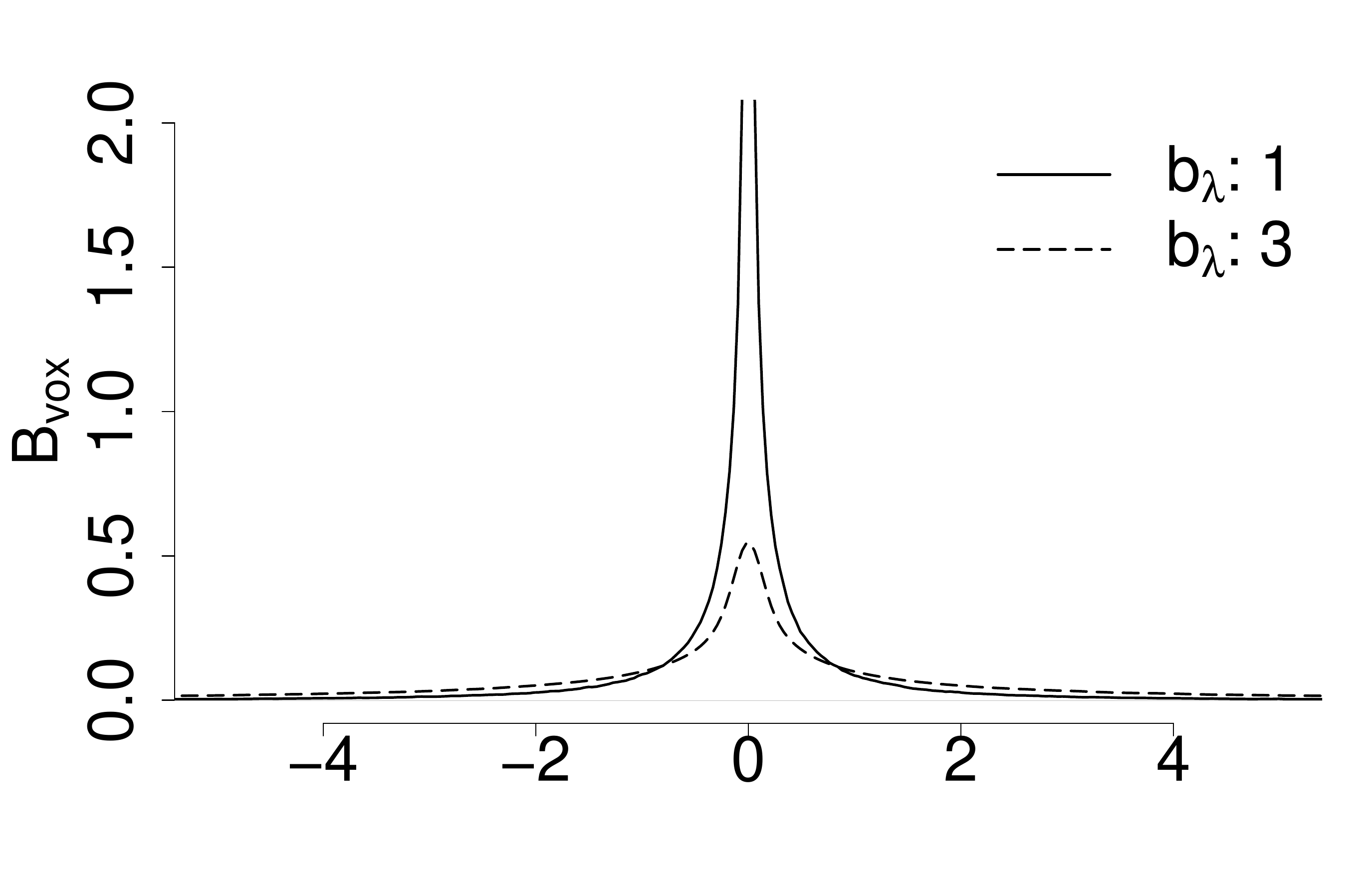} & \includegraphics[width=\linewidth]{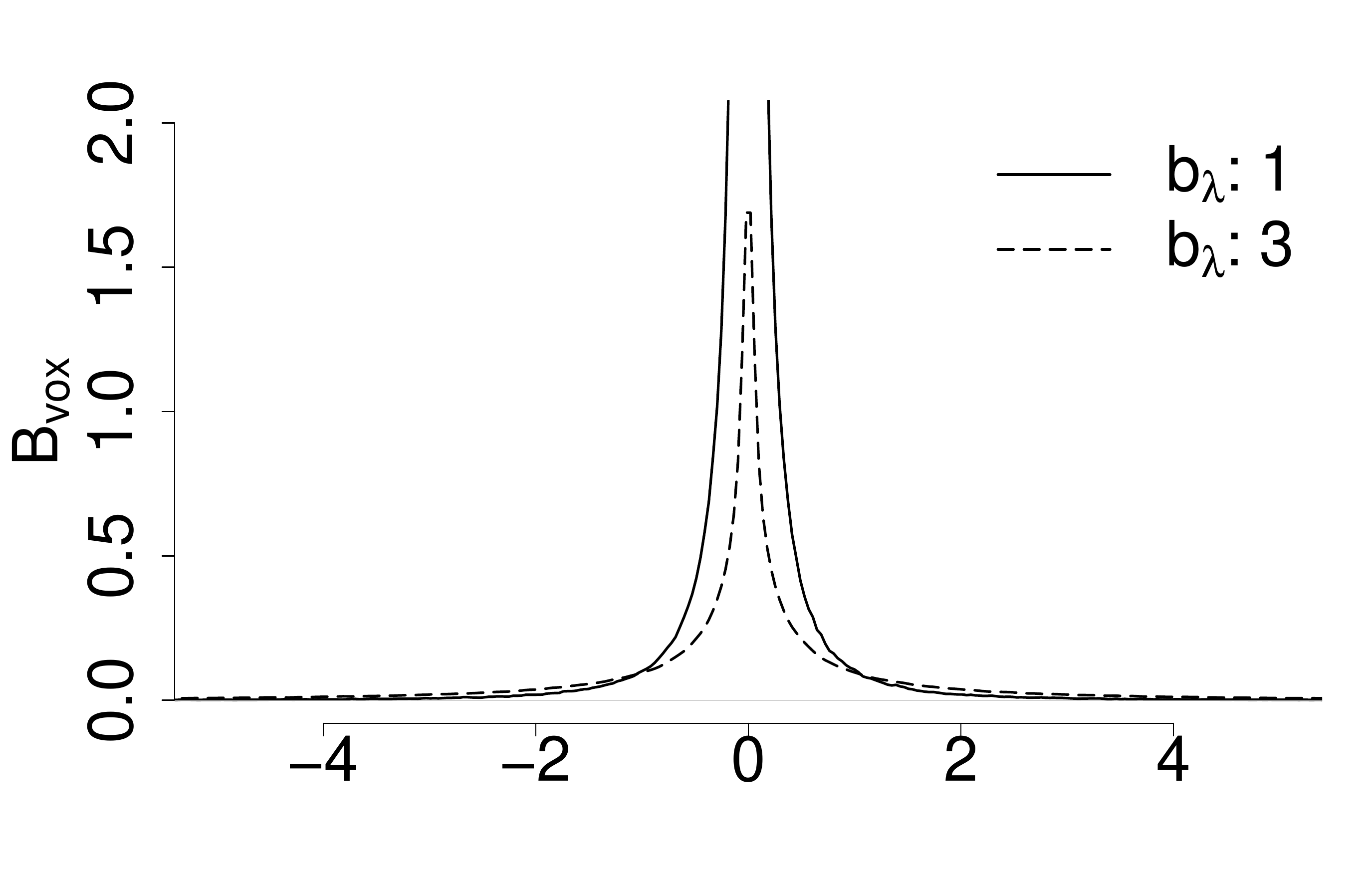}
\end{tabular}
\caption{Induced voxel level prior distribution for default specification for $a_\lambda \in \{2, 3, 5\}$, $b_\lambda \in \{1, 3\}$, with $b_\tau = \alpha R^{1/D}$ and $\alpha = 1/R$. Here, $R = 10$ and $D = 2$; note that $b_\tau = \sqrt[2D]{a_\lambda} \in (1.18,1.50)$ for the range of $a_\lambda$ considered.}
\label{fig:simu_dl_gdp_prior_hyperpar}
\end{figure}



\section{Posterior consistency for tensor regression}
\label{sec:tensor_theory}

\subsection{Notation and framework}
We establish convergence results for tensor regression model \eqref{eq:tensor_reg_model} under the following simplifying assumptions\footnote{Simplifying assumptions are merely to ease notation and calculations; all results generalize in a straightforward manner.}: (i) the intercept is omitted by centering the response; (ii) the error variance is known as $\sigma_0^2 = 1$; and (iii) fixed effects are known as $\bgamma=(0, \dots, 0)$.
We consider an asymptotic setting in which the dimensions of the tensor grow with $n$.  This paradigm attempts to capture the fact that tensor dimension $\prod_j p_{j,n}$ is typically substantially larger than sample size.  This creates theoretical challenges, related to (but distinct from) those faced in showing posterior consistency for high dimensional regression \citep{armagan2013generalized} and multiway contingency tables \citep{zhou2014bayesian}.  

Suppose the data generating model is in the assumed model class \eqref{eq:tensor_reg_model}, i.e., having true tensor parameter $\bB_n^0 \in \otimes_{j=1}^D \Re^{p_j,n}$, error variance $\sigma_0^2=1$ and 
\begin{align*}
\bB_n^0 = \sum_{r=1}^{R} \bbeta_{1,n}^{0(r)}\circ \cdots \circ \bbeta_{D,n}^{0(r)}, \qquad
\bbeta_{j,n}^{0(r)} = (\beta_{j,n,1}^{0(r)}, \dots, \beta_{j,n,p_{j,n}}^{0(r)})' \in \Re^{p_{j,n}}.
\end{align*}
In addition, define $\bF_n, \bF_n^0 \in \Re^{R\sum_{j=1}^{D}p_{j,n}}$ as the vectorized parameters: 
\begin{align*}
&\bF_n = \vec\big(\bbeta_{1,n}^{(1)}, \cdots, \bbeta_{1,n}^{(R)}, \cdots, \bbeta_{D,n}^{(1)}, \cdots,\bbeta_{D,n}^{(R)}\big) \\
&\bF_n^0 = \vec\big(\bbeta_{1,n}^{0(1)}, \cdots, \bbeta_{1,n}^{0(R)}, \cdots, \bbeta_{D,n}^{0(1)},\cdots,\bbeta_{D,n}^{0(R)}\big).
\end{align*}

Define a Kulback-Leibler (KL) neighborhood around the true tensor $\bB_n^0$ as
\[
\mathcal{B}_n=\left\{\bB_n:\frac{1}{n}\sum_{i=1}^n \KL(f(y_i|\bB_n^0),f(y_i|\bB_n))<\epsilon\right\}.\]
Denote $\KL(f(y_i|\bB_n^0),f(y_i|\bB_n))$ as $\KL_i$. Since $\KL_i=\frac{1}{2}\left(\langle \bX_i,\bB_n^0\rangle-\langle \bX_i,\bB_n\rangle\right)^2$, the KL neighborhood of radius $\epsilon$ around $\bB_n^0$ can be rewritten as $\mathcal{B}_n=\big\{\bB_n:\frac{1}{2n}\sum_{i=1}^n \big(\langle \bX_i,\bB_n^0\rangle-\langle \bX_i,\bB_n\rangle\big)^2<\epsilon\big\}$. Further let $\pi_n$ and $\Pi_n$ denote prior and posterior densities with $n$ observations, respectively, and
\begin{align*}
\Pi_n(\mathcal{B}_n^c)=\frac{\int_{\mathcal{B}_n^c}f(\by_n|\bB_n)\pi_n(\bF_n)}{\int f(\by_n|\bB_n)\pi_n(\bF_n)}
\end{align*}
with $\by_n=(y_1,\dots,y_n)'$ and $f(\by_n|\bB_n)$ the density of $\by_n$ under model \eqref{eq:tensor_reg_model}. Posterior consistency is established by showing that
\begin{align}\label{eq:postcons}
\Pi_n\left(\mathcal{B}_n^c\right)\rightarrow 0 \quad \mbox{under}~\bB_n^0\quad\mbox{a.s. as}~n\rightarrow\infty.
\end{align}

\subsection{Main result}

Our main theorem is that (\ref{eq:postcons}) holds under a simple sufficient condition on the prior.
\begin{theorem}\label{theorem:main}
Let $\zeta_n=n^{\frac{1+\rho_3}{2}}$ ($\rho_3>0$), $M_n=\frac{1}{n}\sqrt{\sum\limits_{i=1}^n ||\bX_i||_2^2}$. Given Lemma \ref{lemma:numconc}, for any $\epsilon>0$, $\Pi_n(\bB_n:\frac{1}{n}\sum_{i=1}^n \KL_i>\epsilon)\rightarrow 0$ a.s. under $\bB_n^0$, for the prior
$\pi_n(\bB_n)$ that satisfies
\begin{align}\label{priorprop}
\pi_n\left(\bB_n:||\bB_n-\bB_n^0||_2<\frac{2\eta}{3M_n\zeta_n}\right)>\exp(-dn),\:\:\mbox{for all large}\: n
\end{align}
for any $d>0$ and $\eta<\frac{\epsilon}{32}-d$.
\end{theorem}
Lemma \ref{lemma:numconc} verifies the existence of exponentially-consistent tests.  The proof of the Lemma and theorem are provided in the Appendix.
The proposed multiway shrinkage prior satisfies \eqref{priorprop} and hence leads to posterior consistency under the following theorem.
\begin{theorem}\label{tensorDLGDP}
For fixed constants $H_1, H_2, M_1, \rho_1,$ and $\rho_2 > 0$, the M-DGDP prior \eqref{eq:M1} yields posterior consistency under conditions:
\begin{enumerate}[(a)]
\item $H_1n^{\rho_1}<M_n<H_2n^{\rho_2}$
\item $\sup_{l=1,\dots,p_{j,n}}|\bbeta_{j,n,l}^{0(r)}|<M_1<\infty$, for all $j=1,\dots,D;\:r=1,\dots,R$
\item $\sum_{j=1}^D p_{j,n}\log(p_{j,n})=o(n)$.
\end{enumerate}
\end{theorem}

\begin{remark}
Theorem \ref{tensorDLGDP} require that $\sum_{j=D}p_{j,n}$ grows sub-linearly with sample size $n$. 
 \end{remark}

\section{Posterior computation and model fitting}
\label{sec:posterior_computation}
Letting $y \in \Re$ denote a response, and $\bz \in \Re^p, \bX \in\otimes_{j=1}^D \Re^{p_j}$ predictors, we let
\begin{align}\label{eq:tensor_reg_model}
\begin{aligned}
&y| \bX, \bgamma, \bB, \sigma \sim \mathrm{N}\big(\bz' \bgamma + \langle \bX,\bB \rangle, \sigma^2\big) \\
&\bB = \sum_{r=1}^R \bB_r, ~\bB_r = \bbeta_D^{(r)} \otimes \cdots \otimes \bbeta_1^{(r)} \\
&\sigma^2 \sim \pi_\sigma, ~\bgamma \sim \pi_\gamma, ~\bbeta_j^{(r)} \sim \pi_\bbeta.
\end{aligned}
\end{align}
The noise variance is given a conjugate inverse-gamma prior, $\sigma^2 \sim \mathrm{IG}(v/2, v s_0^2 /2)$, with $s_0^2$ chosen  so that $\mathrm{Pr}(\sigma^2 \le 1) = 0.95$. This is done assuming the response is centered and scaled, which also removes the need for an intercept; by default we set $v = 2$. Fixed effects are given conjugate normal prior $\bgamma \sim \mathrm{N}(0, \sigma^2 \bSigma_{0\bgamma})$, with rescaled prior covariance. Finally, voxel data for the tensor predictor are standardized to have mean zero and variance 1, allowing one to assume default values for hyper-parameters of the proposed multiway priors.

\subsection{Posterior computation}
\label{sec:dl_gdp_mcmc}
The proposed multiway M-DGDP prior \eqref{eq:M1} leads to efficient posterior computation for tensor regression model \eqref{eq:tensor_reg_model}. We rely on marginalization and blocking to reduce auto-correlation for $\big\{\big(\bbeta_j^{(r)}, w_{jr}; 1 \le j \le D, 1 \le r \le R\big), (\Phi, \tau), (\bgamma,\sigma)\big\}$, drawing in sequence from (i) $[\alpha, \Phi, \tau | \bB, W]$; (ii) $[\bB, W | \Phi, \tau, \bgamma, \sigma, \by]$; and (iii) $[\bgamma, \sigma | \bB, \by]$.
Step (i) is non-trivial and we propose an efficient way to sample this block of parameters  compositionally. This is {\em essential} for good  mixing under the M-DGDP prior. Step (ii) is sampled using a sequence of draws from full conditional distributions using a back-fitting procedure to iterate draws from margin-level conditional distributions across the components.

\begin{enumerate}[(1)]
\item Sample $[\alpha, \Phi, \tau | \bB, \bW] = [\alpha | \bB, \bW] [\Phi, \tau | \alpha, \bB, \bW]$;
\begin{enumerate}[(a)]
\item Sample from the conditional distribution of Dirichlet concentration parameter $[\alpha | \bB, \bW]$ via griddy-Gibbs:
form a reference set by drawing $M$ samples from $[\Phi, \tau | \alpha, \bB, \bW]$ for each $\alpha \in \mathcal{A}$. Set $w_{j,l} = \pi(\bB | \alpha, \Phi_l, \tau_l, \bW) \pi(\Phi_l, \tau_l | \alpha)$, $1 \le l \le M |\mathcal{A}|$,  $p(\alpha | \bB, \bW) = \pi(\alpha) \sum_{l=1}^{M |\mathcal{A}|} w_{j,l} / M$, and $\mbox{Pr}(\alpha=\alpha_j|-) = p(\alpha_j | \bB, \bW) /\allowbreak \sum_{\alpha \in \mathcal{A}} p(\alpha | \bB, \bW)$.
\item
Next, the rank-specific scales are sampled (see Appendix \ref{sec:tensor_mcmc}) as $[\Phi, \tau | \alpha^*, \bB, \bW] =  [\Phi | \bB, \bW][\tau | \Phi, \bB, \bW]$; define $p_0 = \sum_{j=1}^D p_j$, and recall $a_\tau = \sum_{r=1}^R \alpha_r = R \alpha$ and $b_\tau = \alpha(R / v)^{1/D}$ (see Section \ref{sec:mdgdp_prior}), then
\begin{itemize}
\item draw $\psi_r\sim \mathrm{giG}(\alpha - p_0 / 2, 2b_\tau, 2C_r)$, $C_r = \sum_{j=1}^D \bbeta_j^{(r)\,T} \bW_{jr}^{-1} \bbeta_j^{(r)}$, and set $\phi_r = \psi_r / \sum_{l=1}^R \psi_l$ in parallel for $1 \le r \le R$
\item draw $\tau \sim \mathrm{giG}(a_\tau - Rp_0 / 2, 2b_\tau, 2\sum_{r=1}^R D_r)$, $D_r = C_r/\phi_r$
\end{itemize}
\end{enumerate}

\item Sample $\big\{(\bbeta_j^{(r)}, w_{jr}, \lambda_{jr}); 1\le j \le D, 1\le r \le R\big\} | \Phi, \tau, \bgamma, \sigma, \by$. For $r = 1, \dots, R$ and $j = 1, \dots, D$, cycle over $[(\bbeta_j^{(r)},  w_{jr}, \lambda_{jr}) | \bbeta_{-j}^{(r)}, \bB_{-r}, \Phi, \tau, \bgamma, \sigma, \by]$, where $\bbeta_{-j}^{(r)} = \{ \bbeta_l^{(r)}, \, l \ne j\}$ and $\bB_{-r} = \bB \setminus \bB_r$;
\begin{enumerate}[(a)]
\item draw $[w_{jr}, \lambda_{jr} | \bbeta_j^{(r)}, \phi_r, \tau] =[w_{jr} | \lambda_{jr}, \bbeta_j^{(r)}, \phi_r, \tau] [\lambda_{jr} | \bbeta_j^{(r)}, \phi_r, \tau]$:
\begin{itemize}
\item draw $\lambda_{jr} \sim \mathrm{Ga}\big(a_\lambda + p_j, b_\lambda + ||\bbeta_j^{(r)}||_1 / \sqrt{\phi_r \tau}\big);$ and
\item draw $w_{jr,k} \sim \mathrm{giG}\big(\frac{1}{2}, \lambda_{jr}^2, \beta_{j,k}^{2\,(r)} / (\phi_r \tau)\big)$ independently for $1 \le k \le p_j$
\end{itemize}

\item draw $\bbeta_j^{(r)} \sim \mathrm{N}(\bmu_{jr}, \bSigma_{jr} )$: define $h_{i,j,k}^{(r)}=\sum_{d_1=1, \dots, d_D=1}^{p_1, \dots, p_D} I(d_j = k) \, x_{d_1, \dots, d_D} \big(\prod_{l \neq j} \beta_{l, i_l}^{(r)}\big)$, $\bH_{i,j}^{(r)} = (h_{i,j,1}^{(r)}, \dots, h_{i,j,p_j}^{(r)})'$, $\tilde y_i = y_i - \bz_i'\bgamma - \sum_{l \neq r} \langle\bX_i,\bB_l\rangle$ for $1 \le i \le n$; then $\bSigma_{jr} = \big( \bH_j^{(r)\, T}\bH_j^{(r)} / \sigma^2 +  \bW_{jr}^{-1} / (\phi_r \tau) \big)^{-1}$, $\bmu_{jr} = \bSigma_{jr} \bH_j^{(r)} \tilde \by / \sigma^2$
\end{enumerate}

\item Sample $[\bgamma, \sigma | \bB, \by] = [\bgamma | \sigma, \tilde\by] [\sigma^2 | \tilde \by]$; define $\tilde y_i =  y_i - \langle\bX_i, \bB\rangle$ for $1 \le i \le n$, then
\begin{enumerate}[(a)]
\item draw $\sigma^2 \sim \mathrm{IG}(a_\sigma, b_\sigma)$, $a_\sigma = (n + v)/2$, $b_\sigma = \big(vs_0^2 + \norm{\tilde \by}_2^2 - \tilde\by^T \bZ \bmu_\bgamma\big)/2$
\item draw $\bgamma \sim \mathrm{N}\big(\bmu_\bgamma, \sigma^2 \bSigma_\bgamma\big)$, $\bSigma_\bgamma = \big(\bZ^T\bZ + \bSigma_{0\bgamma}^{-1}\big)^{-1}$, $\bmu_\bgamma = \bSigma_\bgamma \bZ^T \tilde \by$.
\end{enumerate}
\end{enumerate}
\section{Simulation studies}
\label{sec:simulation_studies}

To illustrate finite-sample performance of the proposed multiway priors, we show results from a simulation study with various dimensionality $(p, R)$ and define $\bar{b} = \max |\bB^0_{i_1, \dots, i_D}|$ as the maximum signal size. Throughout, set $p_j = p$, $\sigma_0^2 = 1$ and $\bar{b} = 1$ for convenience. In addition, we set $\bgamma_0 = (0, \dots, 0)$ and focus exclusively on inference for tensor parameter $\bB$.
The following simulated setups are considered:
\begin{enumerate}[1.]
\item ``Generated'' tensor: We construct tensor parameters having rank $R_0 = \{3,5\}$ with $p = \{64, 100\}$ and $D = 2$. 
\item ``Ready made'' tensor: We use three tensor (2D) images without generating them from a parafac decomposition with known rank. 
\end{enumerate}
Five replicated datasets with $n = 1000$ are generated according to \eqref{eq:tensor_reg_model} with $x_{i_1, \dots, i_D} \sim \mathrm{N}(0, 1)$. The tensor parameters considered are shown in Figure \ref{fig:2d_images}, where the magnitude of the non-zero voxels is $\bar{b} = 1$. Examples are chosen to demonstrate recovery of voxel-level coefficients across varying degrees of complexity (dimension,  parafac rank) and sparsity (\% of non-zero voxels; see Figure \ref{fig:2d_images}).  The performance of our method with M-DGDP prior \eqref{eq:M1} is compared with (i) frequentist tensor regression (FTR)\citep{zhou2013tensor}; and (ii) Lasso (on the vectorized tensor predictor). Comparisons  are based on (a) voxel mean squared estimation error (true non-zero, true zero, and overall); and (b) frequentist coverage (and length) of 95\% credible intervals. MCMC was run for 1000 iterations, with a 200 iteration burn-in and remaining samples thinned by 5. FTR uses $R=5$, selecting the tuning parameter over a grid of values to minimize RMSE. To choose initial values, a preliminary analysis was run with a coarsened $16 \times 16$ image\footnote{On several of the simulated examples, this seems to help optimization, both in terms of runtime as well as in terms of RMSE at convergence.}. In all simulation experiments, we observed rapid mixing for the proposed MCMC algorithm.

Voxel-level RMSE reported in Table \ref{tab:2d_simulation_stat} demonstrates that our method (M-DGDP) consistently out performs FTR. When the tensor parameter has a low-rank parafac decomposition (`R3-ex' and `R5-ex'), M-DGDP and FTR perform best, with M-DGDP having lower RMSE on both true zero and non-zero voxels.  This validates empirically prior \eqref{eq:M1} along with our suggested default hyper-parameter choices (see Section \ref{sec:multprior}); in particular, M-DGDP adapts to varying degrees of sparsity, shrinking many coefficients close to zero while accurately estimating nonzero voxels. FTR's performance is sensitive to the performance of cross validation for parameter tuning, with the CV grid sensitive to tensor dimension ($p,D$) and rank $R$.  In some cases, overall RMSE was lower for $R=3$ even though performance in estimating non-zero parameters was worse than for other choices. 

Results in Table \ref{tab:2d_simulation_mdgdp_coverage} demonstrate that M-DGDP yields credible intervals with good frequentist coverage across each of the simulated settings, both overall as well as on the true non-zero coefficients. Our method is one of the first to offer uncertainty quantification for tensor valued predictors, of critical performance in performing inferences on these parameters.
Finally, Table \ref{tab:2d_simulation_mdgdp_dimension} provides evidence of the robustness of our method to increasing predictor dimension using two of the simulated examples. In both cases, RMSE for FTR worsens considerably on the true zero coefficients. For the true nonzero voxels, RMSE increases for both methods as the margin dimension increases; on a relative \% basis, however, FTR worsens considerably more, while on an absolute scale, M-DGDP remains the clear winner.


\begin{figure}[!h]
\centering
\setlength{\tabcolsep}{1pt}
\begin{tabular}{ccc}
\subfigure[R3-ex (7.0\%)]{\includegraphics[width=0.32\linewidth]{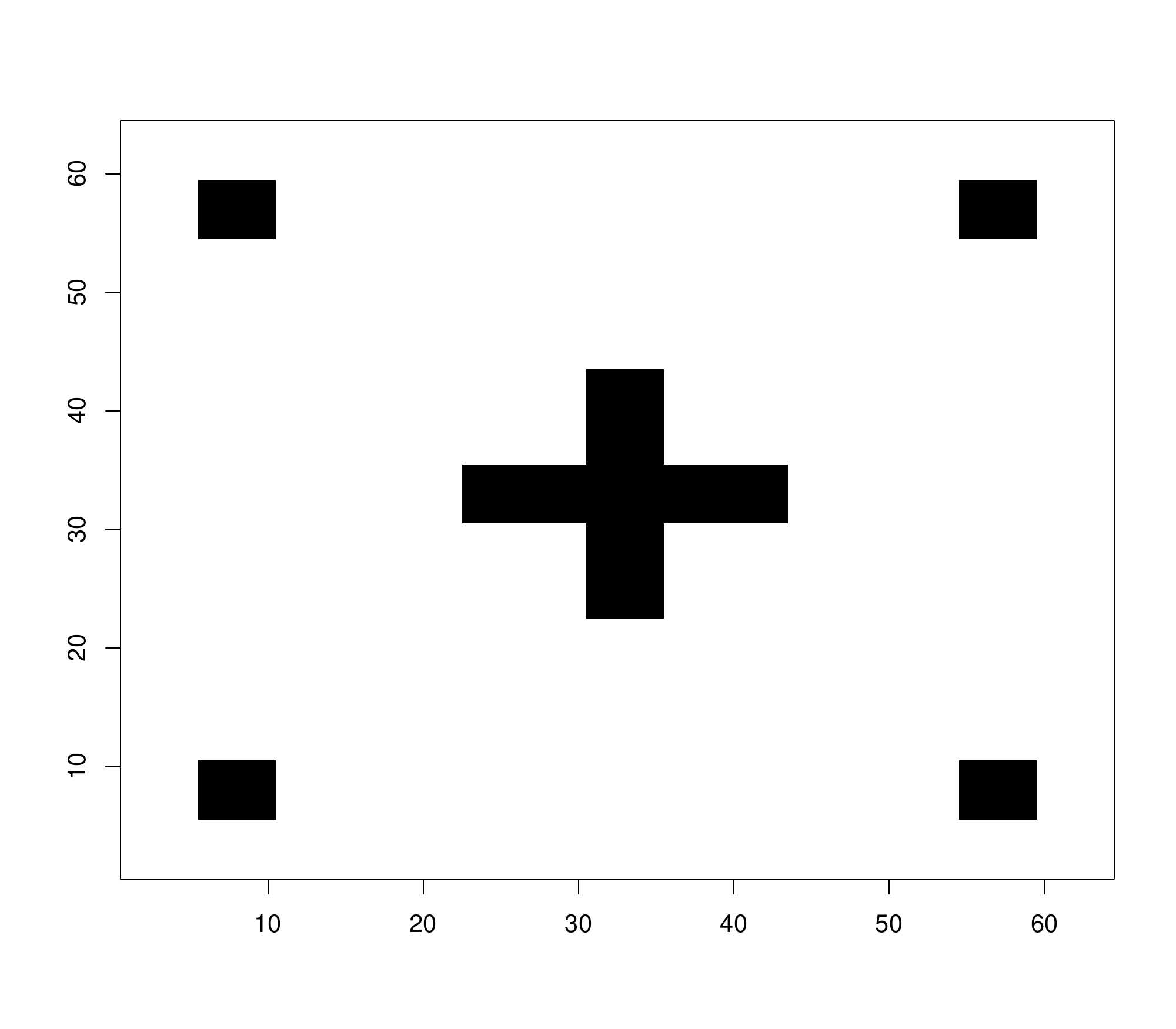}} & \subfigure[R5-ex (11.0\%)]{\includegraphics[width=0.32\linewidth]{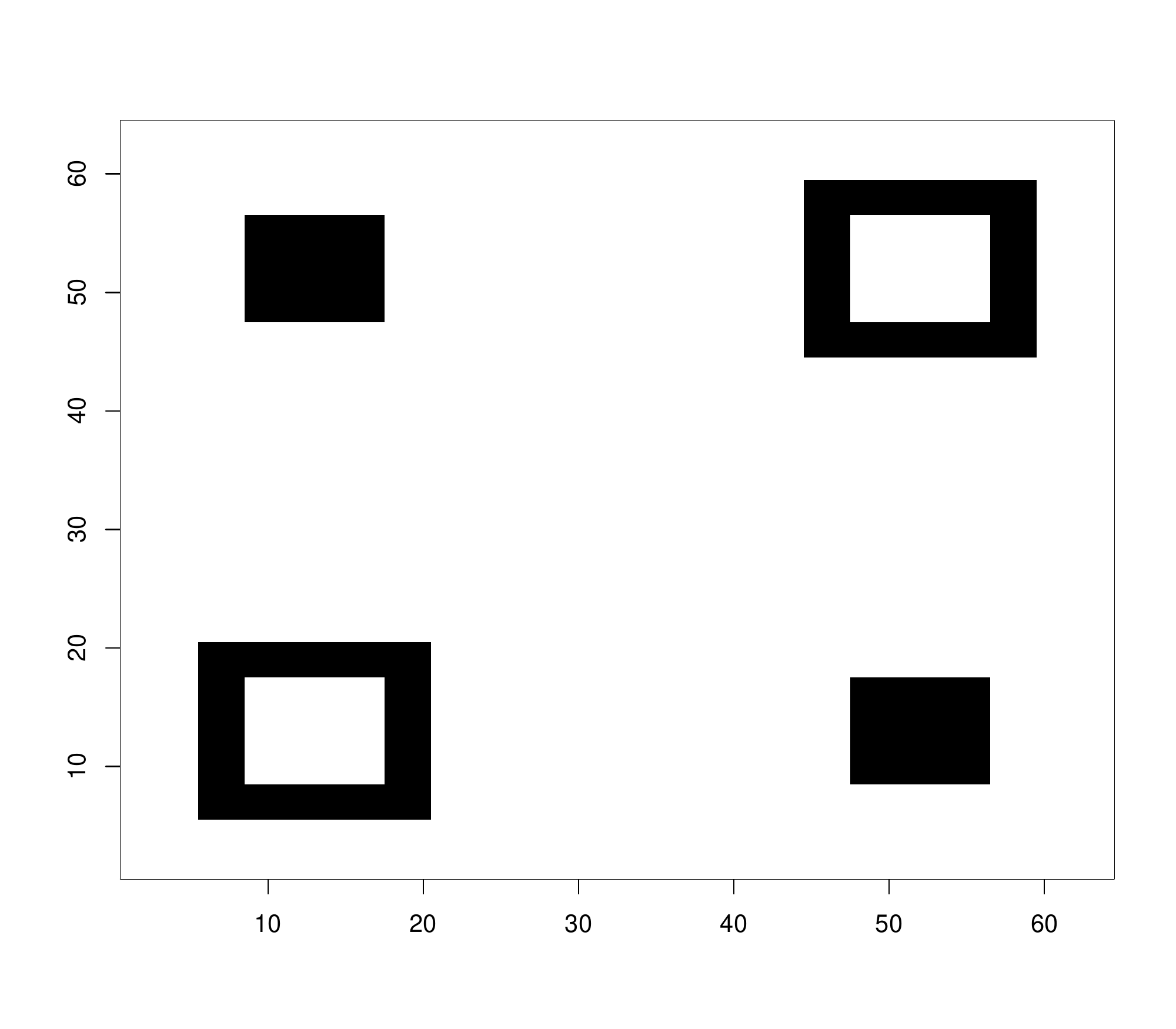}} & \subfigure[Shapes (6.8\%)]{\includegraphics[width=0.32\linewidth]{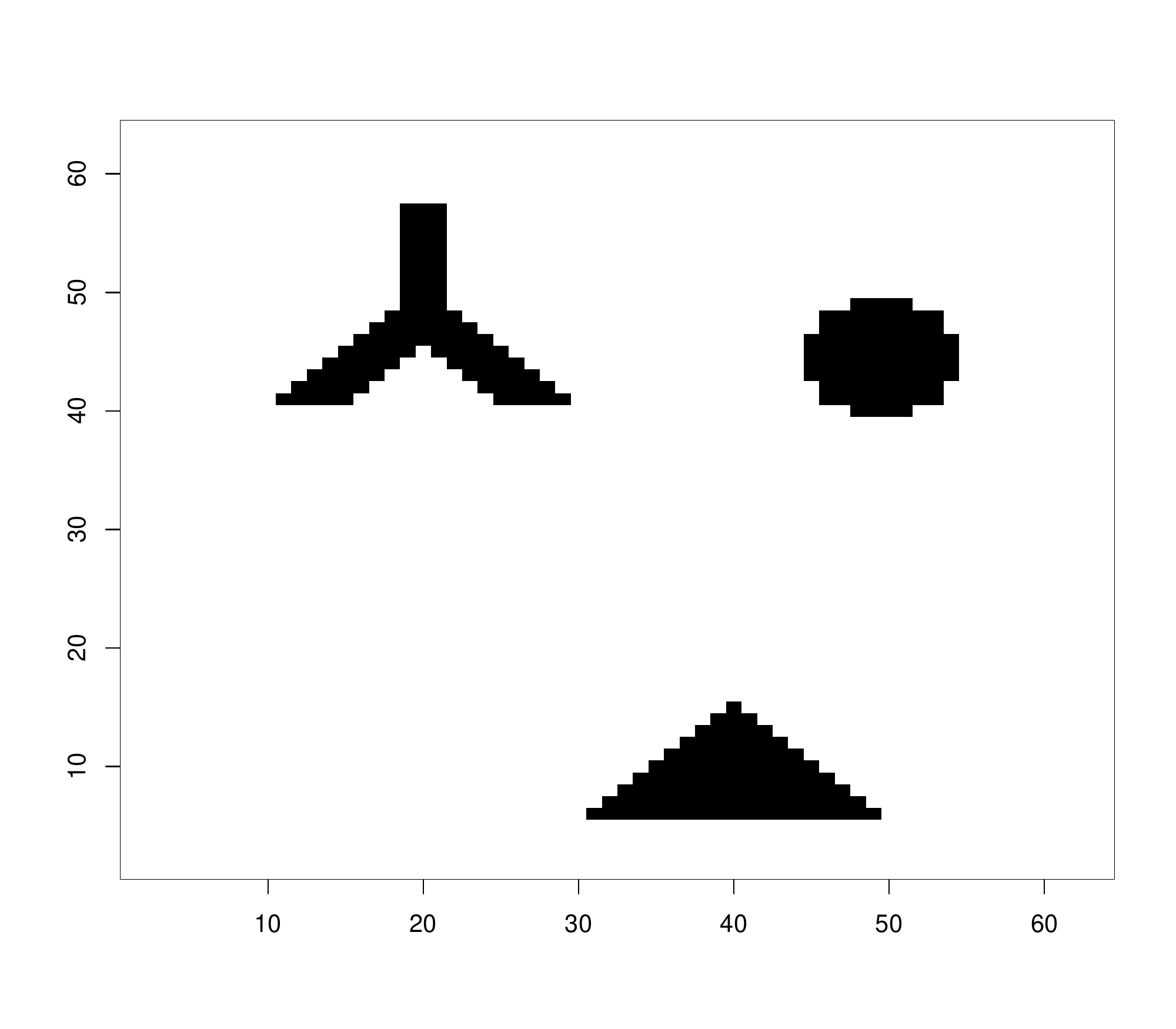}} \\
\subfigure[Eagle (7.4\%)]{\includegraphics[width=0.32\linewidth]{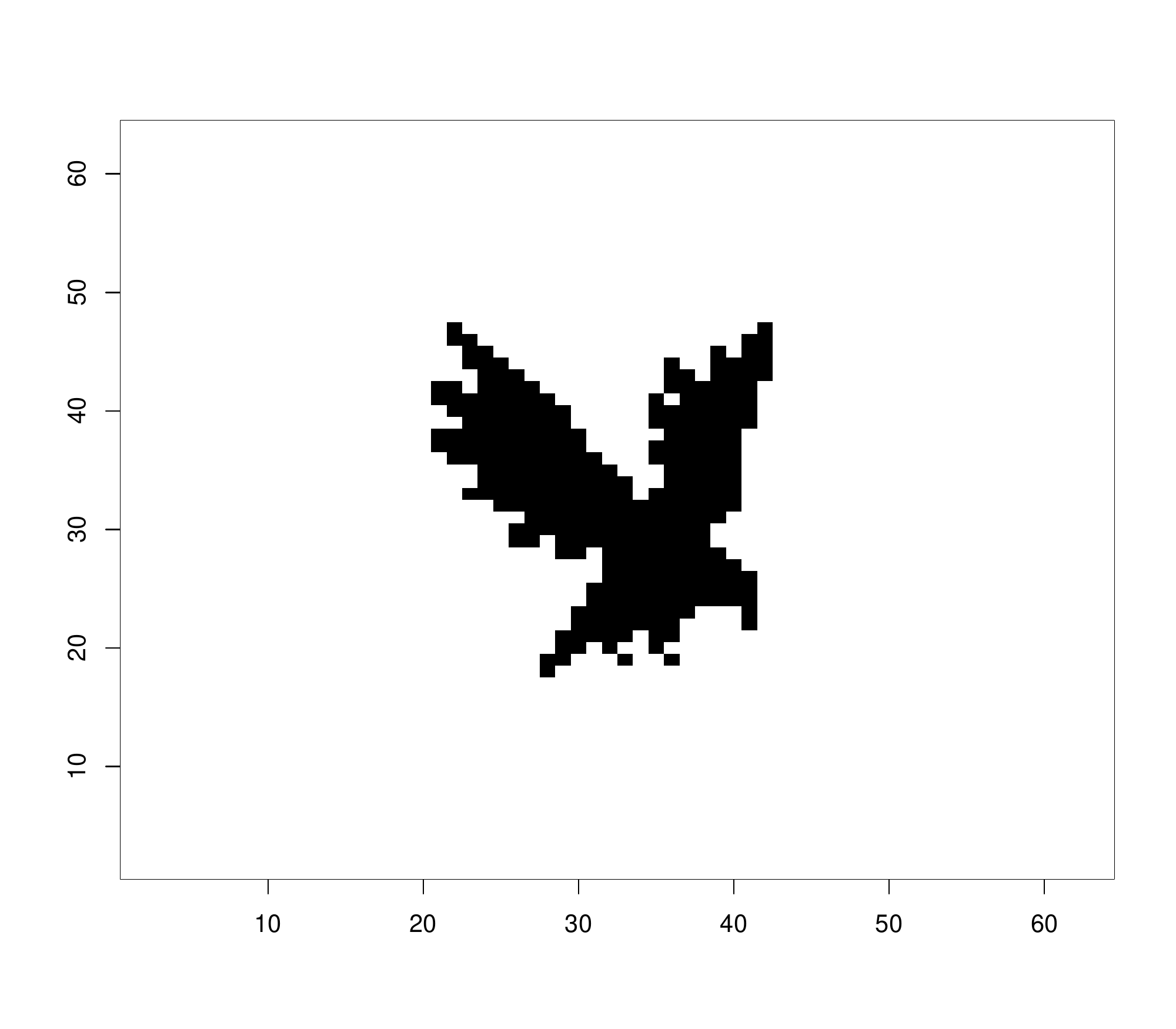}} & \subfigure[Palmtree (8.0\%)]{\includegraphics[width=0.32\linewidth]{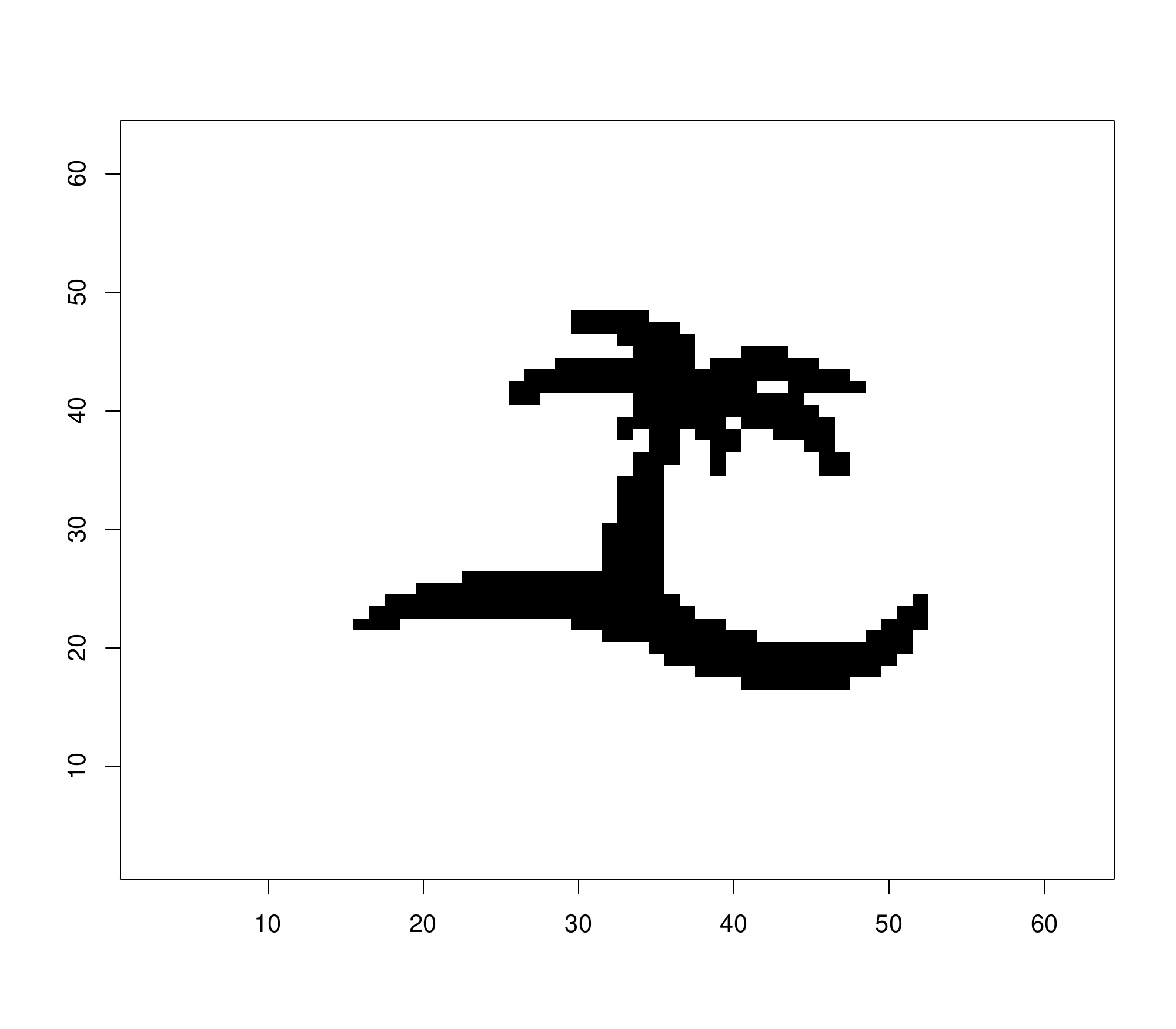}} & \subfigure[Horse (18.6\%)]{\includegraphics[width=0.32\linewidth]{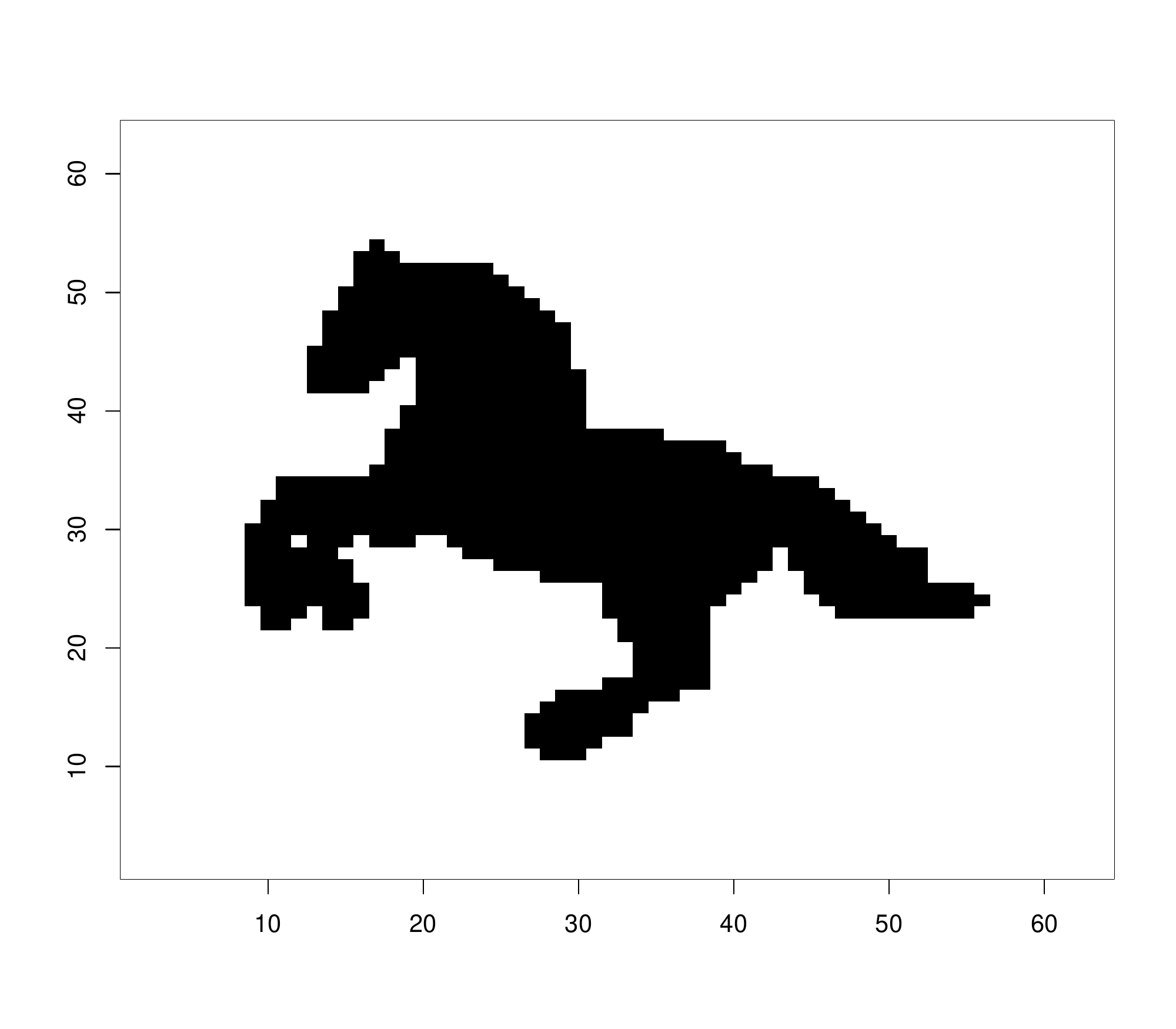}}
\end{tabular}
\caption{Simulated data with $64 \times 64$ 2D tensor images ($p = 64, \, D = 2$). Row 1: The first two images (from left) have a rank-3 and rank-5 parafac decomposition; the third image is ``regular'', although does not have a low-rank parafac decomposition. Row 2: All three images are irregular, and do not have a low-rank parafac decomposition. Sparsity (\% non-zero voxels) are displayed in sub-captions.}
\label{fig:2d_images}
\end{figure}

\begin{figure}[!h]
\centering
\begin{tabular}{ccc}
\subfigure[R3-ex]{\includegraphics[width=0.32\linewidth]{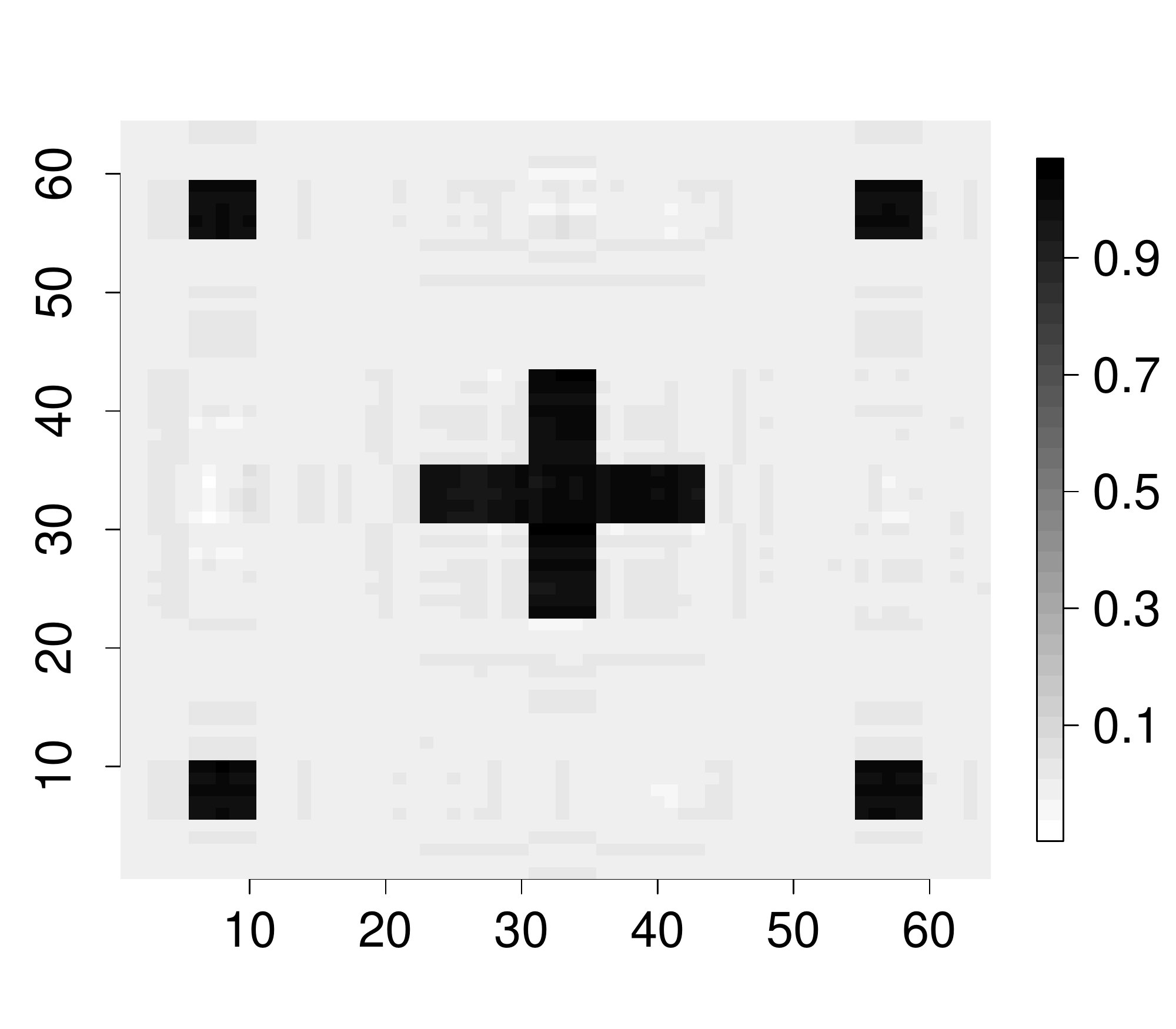}} & \subfigure[R5-ex]{\includegraphics[width=0.32\linewidth]{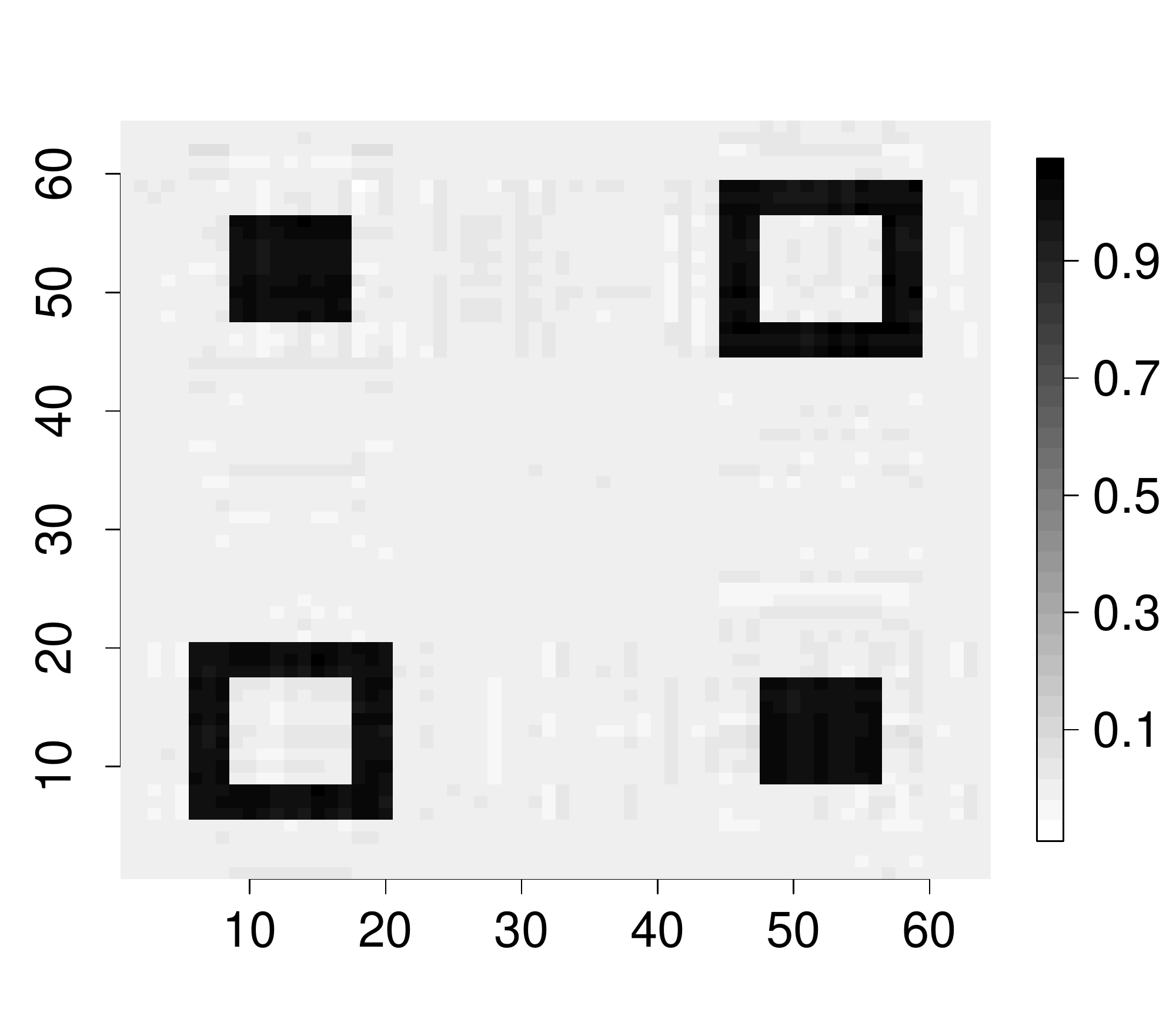}} & \subfigure[Shapes]{\includegraphics[width=0.32\linewidth]{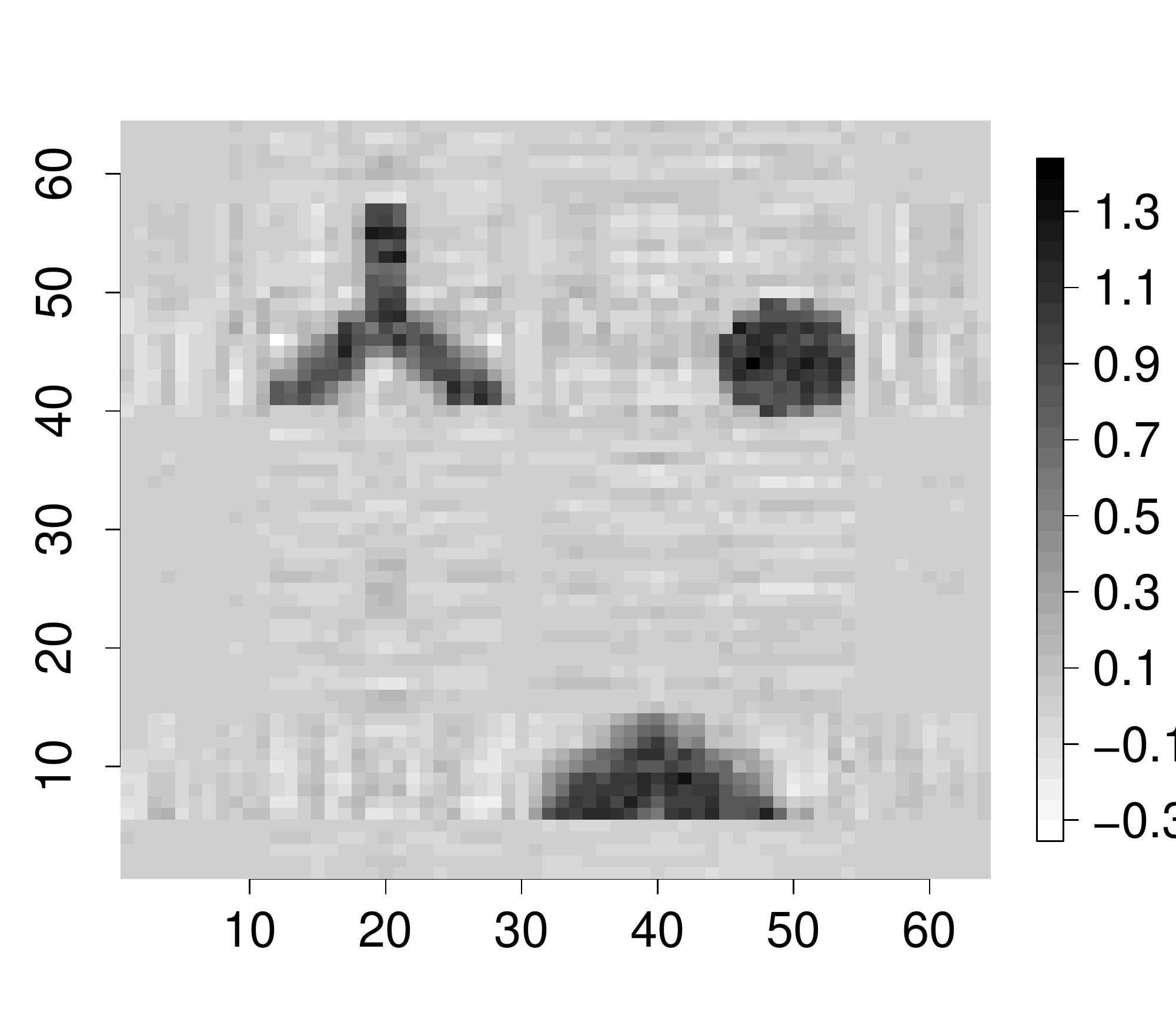}} \\
\subfigure[Eagle]{\includegraphics[width=0.32\linewidth]{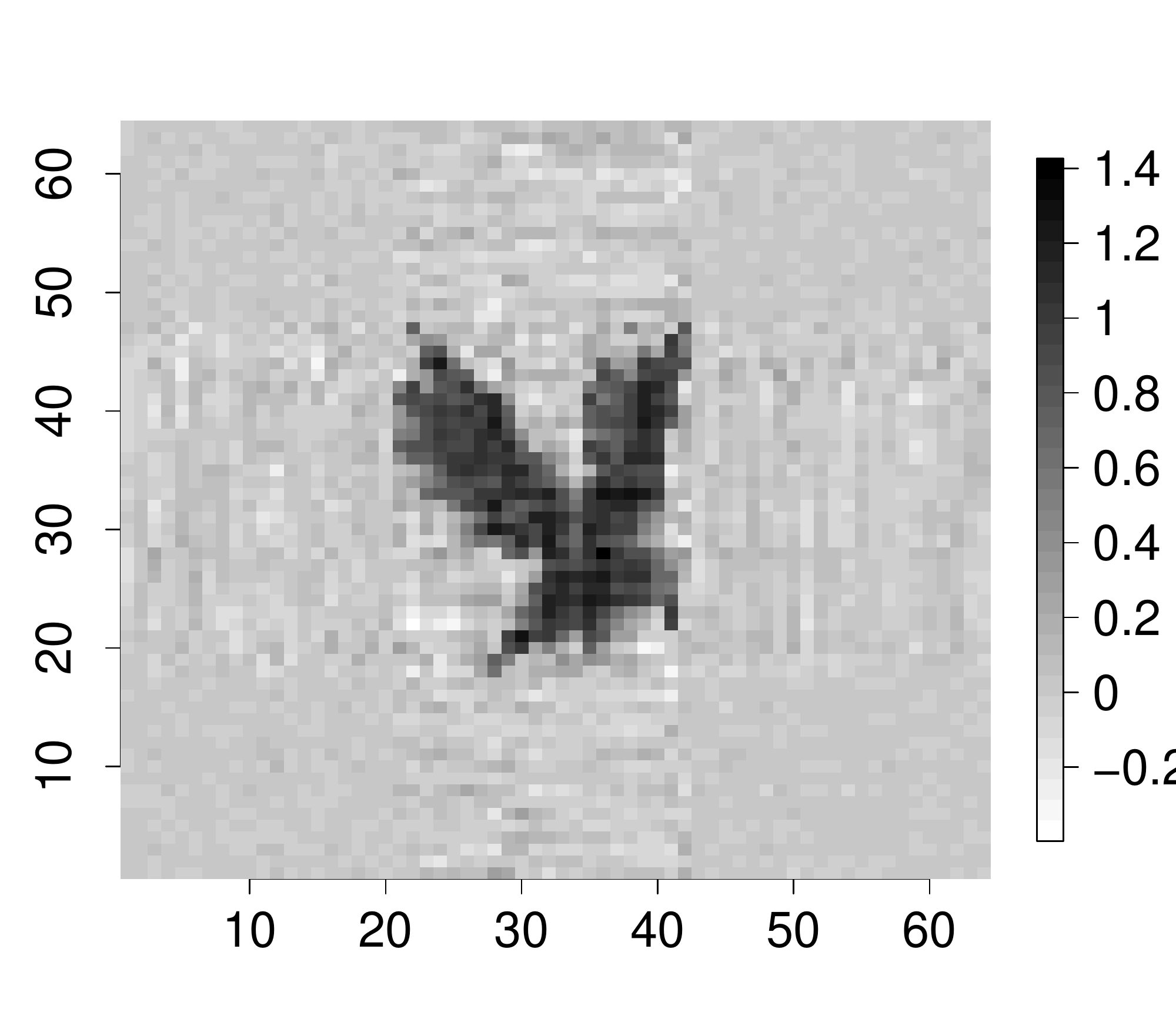}} & \subfigure[Palmtree]{\includegraphics[width=0.32\linewidth]{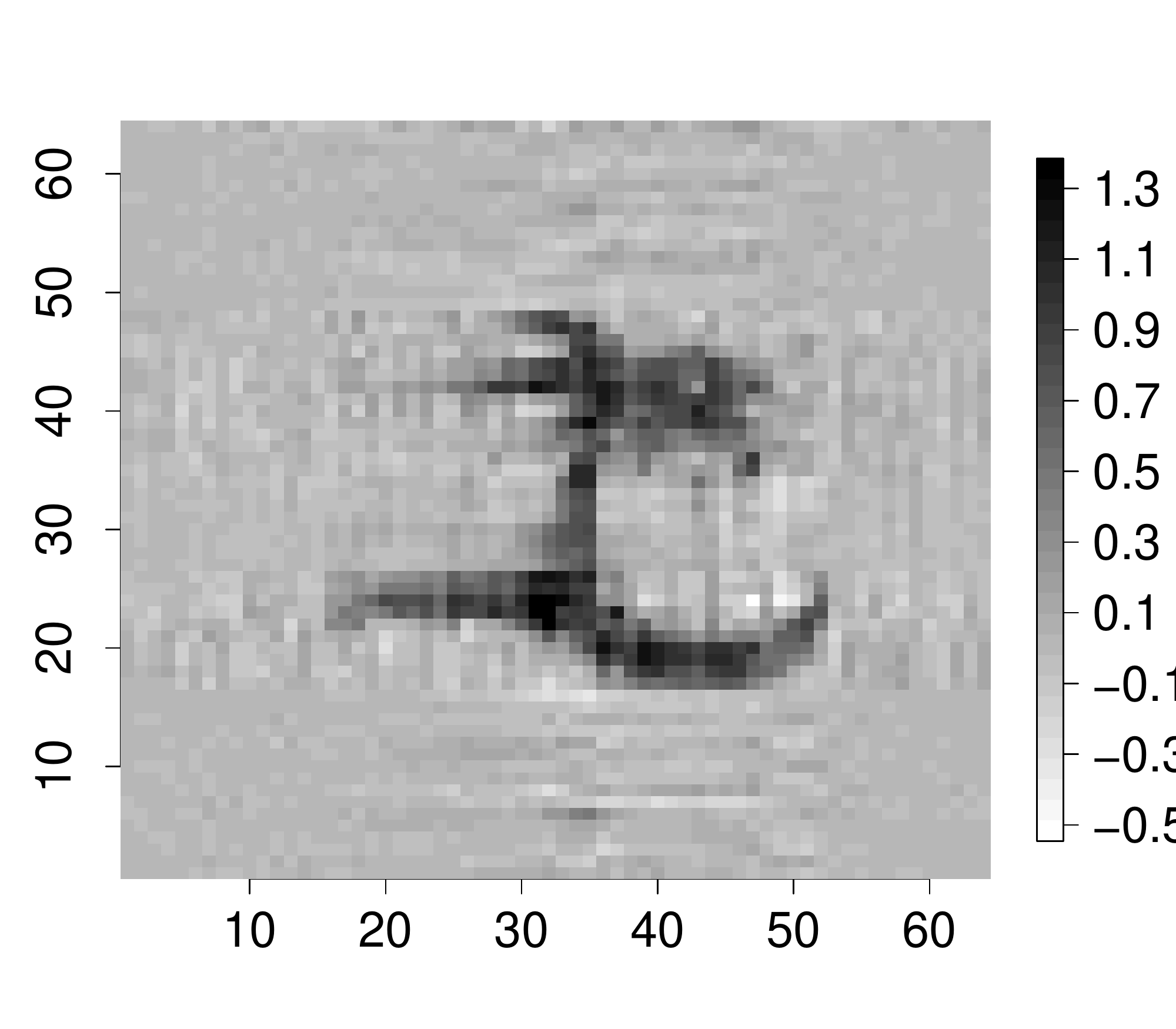}} & \subfigure[Horse]{\includegraphics[width=0.32\linewidth]{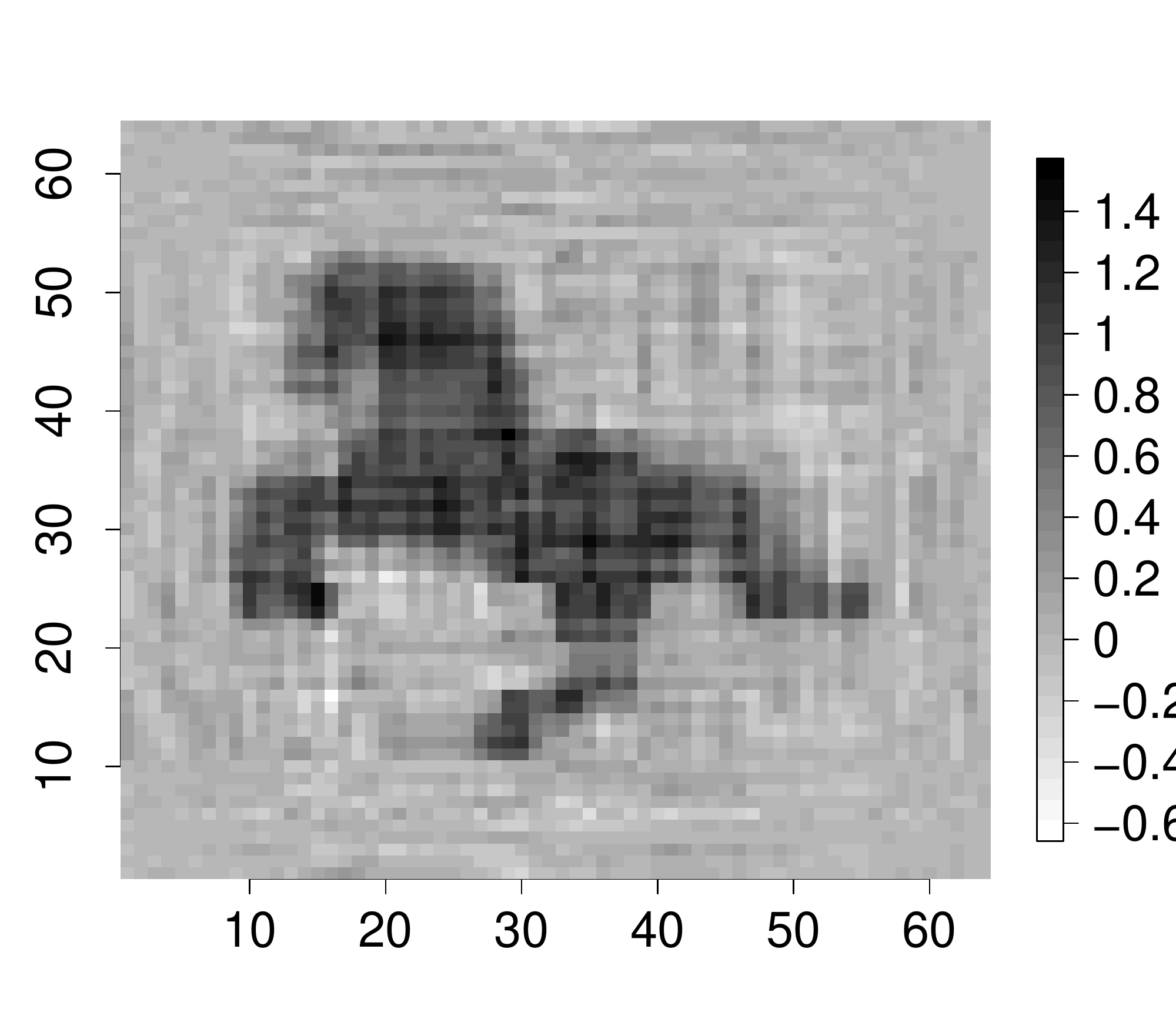}}
\end{tabular}
\caption{Recovered images for the $64 \times 64$ 2D tensor images in Figure \ref{fig:2d_images}. By default, $R = 10$ is used in all M-DGDP runs.}
\label{fig:2d_recover_images}
\end{figure}

\begin{table}[!h]
\footnotesize\centering
\begin{tabular}{l l | cccccc }
& & R3-ex & R5-ex & Shapes & Eagle & Palmtree & Horse \\
\hline 
\multirow{3}{*}{$|\mathrm{vox}_0| > 0$}
& M-DGDP
& \bf 0.023$_{0.00}$ & \bf 0.021$_{0.00}$ & \bf 0.243$_{0.01}$ & \bf 0.226$_{0.02}$ & \bf 0.316$_{0.01}$ & \bf 0.278$_{0.01}$ \\
& FTR 
& 0.035$_{0.00}$ & 0.030$_{0.00}$ & 0.415$_{0.03}$ & 0.354$_{0.03}$ & 0.435$_{0.02}$ & 0.391$_{0.03}$ \\
& Lasso 
& 0.628$_{0.02}$ & 0.822$_{0.03}$ & 0.619$_{0.07}$ & 0.665$_{0.03}$ & 0.698$_{0.03}$ & 0.888$_{0.01}$ \\
\hline
\multirow{3}{*}{$|\mathrm{vox}_0| = 0$}
& M-DGDP 
& \bf 0.011$_{0.00}$ & \bf 0.014$_{0.00}$ & \bf 0.071$_{0.00}$ & \bf 0.085$_{0.00}$ & \bf 0.100$_{0.01}$ & \bf 0.137$_{0.00}$ \\
& FTR 
& 0.022$_{0.00}$ & 0.020$_{0.00}$ & 0.127$_{0.02}$ & 0.163$_{0.03}$ & 0.159$_{0.00}$ & 0.215$_{0.02}$ \\
& Lasso
& 0.090$_{0.00}$ & 0.098$_{0.02}$ & 0.081$_{0.01}$ & 0.097$_{0.00}$ & 0.094$_{0.01}$ & 0.155$_{0.02}$ \\
\hline
\multirow{3}{*}{Overall}
& M-DGDP
& \bf 0.013 & \bf 0.015 & \bf 0.093 & \bf 0.102 & \bf 0.131 & \bf 0.172 \\
& FTR
& 0.023 & 0.021 & 0.164 & 0.184 & 0.196 & 0.257 \\
& Lasso
& 0.187 & 0.288 & 0.179 & 0.204 & 0.217 & 0.407
\end{tabular}
\caption{Comparison of voxel estimation as measured by root mean squared error (RMSE) for the six 2D tensor images portrayed in Figure \ref{fig:2d_images}. Results from FTR \citep{zhou2013tensor} use $R = 5$. By default, $R = 10$ is used in all M-DGDP runs.} 
\label{tab:2d_simulation_stat}
\end{table}

\begin{table}[!h]
\small\centering
\begin{tabular}{l l | cccccc }
& & R3-ex & R5-ex & Shapes & Eagle & Palmtree & Horse \\
\hline 
\multirow{1}{*}{$|\mathrm{vox}_0| > 0$}
& coverage & 0.986$_{0.02}$ & 0.946$_{0.02}$ & 0.747$_{0.01}$ & 0.731$_{0.04}$ & 0.677$_{0.04}$ & 0.795$_{0.02}$ \\
\hline
\multirow{2}{*}{Overall}
& coverage & 0.995$_{0.01}$ & 0.970$_{0.01}$ & 0.965$_{0.00}$ & 0.940$_{0.02}$ & 0.948$_{0.02}$ & 0.927$_{0.01}$ \\
& length & 0.066$_{0.01}$ & 0.061$_{0.01}$ & 0.290$_{0.00}$ & 0.301$_{0.03}$ & 0.410$_{0.03}$ & 0.566$_{0.02}$ \\
\end{tabular}
\caption{M-DGDP coverage statistics on generated and ready-made 2D tensor images with simulated tensor predictor data.}
\label{tab:2d_simulation_mdgdp_coverage}
\end{table}

\begin{table}[!h]
\centering
\begin{tabular}{l l | c | cc | cc }
& & \multirow{2}{*}{$|\mathrm{vox}_0|$} & \multicolumn{2}{c |}{R5-ex} & \multicolumn{2}{c}{Shapes} \\
& & & 64 & 100 & 64 & 100 \\
\hline
\multirow{4}{*}{M-DGDP}
& coverage & $> 0$ 
& 0.946$_{0.02}$ & 0.991$_{0.01}$ &  0.747$_{0.01}$ & 0.590$_{0.06}$ \\
& length & $ > 0$ 
& 0.061$_{0.01}$ & 0.069$_{0.01}$ & 0.290$_{0.00}$ & 0.247$_{0.01}$ \\
& rmse & $> 0$
& 0.021$_{0.00}$ & 0.032$_{0.01}$ & 0.243$_{0.01}$ & 0.320$_{0.03}$ \\
& rmse & $= 0$ 
& 0.014$_{0.00}$ & 0.014$_{0.00}$ & 0.071$_{0.00}$ & 0.063$_{0.00}$ \\
\hline
\multirow{2}{*}{FTR} 
& rmse & $> 0$
& 0.030$_{0.00}$ & 0.369$_{0.06}$ & 0.415$_{0.03}$ & 0.586$_{0.14}$ \\
& rmse & $= 0$
& 0.020$_{0.00}$ & 0.111$_{0.02}$ & 0.127$_{0.02}$ & 0.135$_{0.02}$
\end{tabular}
\caption{Sensitivity analysis of voxel estimation error (RMSE) as the tensor dimension increases; here $p_j = p \in \{64, 100\}$ for the 2D tensor images `R5-ex' and `Shapes'.}
\label{tab:2d_simulation_mdgdp_dimension}
\end{table}

\section{Simulated response with a real 3D brain image}\label{sec:real_3d_image_analysis}
We analyze data containing 3D MRI images for $550$ adolescents, with information such as age and sex available. Age and sex are treated as ordinary scalar covariates while 3D MRI images act as tensor covariates. Let $\bX$ denote a $30\times 30\times 30$ 3D MRI image, $Z_1$ be the age and $Z_2$ be the sex of an individual. The response is simulated using $y\sim \mathrm{N}\left(\bZ'\bgamma+\langle \bX, \bB_0\rangle,\sigma^2\right)$, where $\bZ$ denotes $(Z_1,Z_2)'$, $\bgamma\in\mathcal{R}^2$ and $\bB_0\in \mathcal{R}^{30\times 30\times 30}$. We assume the true $\bB_0$ is a rank-2 tensor, $\bB_0=\bb_1\circ\bb_2\circ\bb_3+\ba_1\circ\ba_2\circ\ba_3$. Initialization and standardization of predictors follow exactly as prescribed in Section \ref{sec:posterior_computation}. 

The following  cases are examined by varying $\ba_i$'s and $\bb_i$'s:
\newline
\emph{Case 1}:
$\bb_1=\bb_2=\left(0,\dots,0,\sin((1:15)*\pi/4)\right)$, $\bb_3=\left(\sin((1:10)*\pi/4),0,\dots,0\right)$,\\
$\ba_1=\left(0,\dots,0,\sin((1:10)*\pi/4)\right)$, $\ba_2=\left(0,\dots,0,\cos((1:15)*\pi/4)\right)$, \\$\ba_3=\left(\sin((1:15)*\pi/4),0,\dots,0\right)$.\\
\emph{Case 2}:
$\bb_1=\bb_2=\left(0,\dots,0,\sin((1:15)*\pi/6)\right)$, $\bb_3=\left(\sin((1:20)*\pi/6),0,\dots,0\right)$,\\
$\ba_1=\left(0,\dots,0,\sin((1:15)*\pi/4)\right)$, $\ba_2=\left(0,\dots,0,\cos((1:10)*\pi/6)\right)$,\\
$\ba_3=\left(\sin((1:15)*\pi/6),0,\dots,0\right)$.\\
\emph{Case 3}:
$\bb_1=\bb_2=\left(0,\dots,0,\sin((1:20)*\pi/6)\right)$, $\bb_3=\left(\sin((1:20)*\pi/6),0,\dots,0\right)$,\\
$\ba_1=\left(0,\dots,0,\sin((1:10)*\pi/4)\right)$, $\ba_2=\left(0,\dots,0,\cos((1:20)*\pi/4)\right)$,\\
$\ba_3=\left(\sin((1:20)*\pi/6),0,\dots,0\right)$.

These cases correspond to sparse $\bB_0$, with 12\%, 18\% and 30\% nonzero elements, respectively.  We implement M-DGDP, FTR with $R=5$ fixed, and Lasso with tensor vectorized.  Both FTR and \cite{zhou2013tensor} include an $L_1$ penalty (results are shown for the best choice of penalty), which can over-shrink voxel coefficients significantly different from zero.  M-DGDP instead includes heavy tail to prevent such over-shrinkage. This is evident from the better performance of M-DGDP prior in terms of estimating nonzero coefficients, see Table 5.  Table~\ref{tab:3d_data} summarizes RMSEs for the estimated tensor coefficients for each of the competitors. In each of the above simulations, trace plots for several model parameters were monitored and found to mix well using the proposed MCMC algorithm in Section \ref{sec:posterior_computation}.
\begin{table}[!h]
\centering
\begin{tabular}{l l | ccc }
& & Case 1 & Case 2 & Case 3 \\
\hline
\multirow{3}{*}{$|\mathrm{vox}_0| > 0$}
& M-DGDP
& \bf 0.39 & \bf 0.30 & \bf 0.34 \\
& FTR
& 0.46 & 0.41 & 0.43\\
& Lasso
& 0.46 & 0.42 & 0.44 \\
\hline
\multirow{3}{*}{$|\mathrm{vox}_0| = 0$}
& M-DGDP
& \bf 0.04 & \bf 0.14 & \bf 0.10  \\
& FTR
& 0.00 & 0.00 & 0.00 \\
& Lasso
& 0.01 & 0.03 & 0.02 \\
\hline
\multirow{3}{*}{Overall}
& M-DGDP
& \bf 0.13 & \bf 0.20 & \bf 0.17 \\
& FTR
& 0.15 & 0.22 & 0.18 \\
& Lasso
& 0.15 & 0.23 & 0.18
\end{tabular}
\caption{Comparison of voxel estimation as measured by root mean squared error (RMSE) for the coefficients in case 1,2, 3 corresponding to 3D tensor images. Results from both M-DGDP and FTR \citep{zhou2013tensor} use $R = 5$.} 
\label{tab:3d_data}
\end{table}
Overall, M-DGDP prior performs $10-15$\% better for cases considered in this section. In less sparse cases, it is also evident that M-DGDP tends to  outperform $L_1$-optimized methods by a greater margin. Importantly, every parameter in M-DGDP is auto-tuned, while $L_1$ penalty results in vastly different performance with varying choices of the tuning parameter.
\begin{table}[!h]
\centering
\begin{tabular}{l  | ccc }
 & Case 1 & Case 2 & Case 3 \\
\hline
Coverage & 0.98 & 0.96 & 0.99\\
\hline
Length & 0.54 & 0.87 & 2.16
\end{tabular}
\caption{Coverage and length for 95\% credible intervals for M-DGDP.} 
\label{tab:coverage_length}
\end{table}

\begin{table}[!h]
\centering
\begin{tabular}{l l | ccc }
& & Case 1 & Case 2 & Case 3 \\
\hline
\multirow{2}{*}{$\gamma_1:$ truth $=$ 0.5}
& M-DGDP
& \bf 0.57 & \bf 0.54 & \bf 0.33 \\
& FTR
& 0.46 & 0.85 & 0.95\\
\hline
\multirow{2}{*}{$\gamma_2:$ truth $=$ 2}
& M-DGDP
& \bf 2.00 & \bf 2.04 & \bf  1.86 \\
& FTR
& 1.87 & 0.22 & 3.30 \\
\end{tabular}
\caption{Point estimates of the coefficients for age and sex for M-DGDP and \citep{zhou2013tensor} along with the true values.} 
\label{tab:vec_coef_data}
\end{table}

While M-DGDP consistently shows coverage over 95\% with reasonably short credible intervals (see Table \ref{tab:coverage_length}),
$L_1$-optimization based methods generally suffer in this regard. For  completeness, we provide point estimates and credible intervals for coefficients corresponding to age and sex in Table~\ref{tab:vec_coef_data}. The data analysis reveals superior performance of M-DGDP with proper characterization of uncertainties.

\section*{Appendix}\label{sec:appendix_tensor}

\subsection*{MCMC algorithm}
\label{sec:tensor_mcmc}

The following derivations concern the M-DGDP prior \eqref{eq:M1} and the sampling algorithm outlined in Section \ref{sec:dl_gdp_mcmc}.

\paragraph{For step (1b)}
Recall from Section \ref{sec:mdgdp_prior} that $\tau \sim \mathrm{Ga}(a_\tau, b_\tau)$ and $\Phi \sim \mathrm{Dirichlet}(\alpha_1, \dots, \alpha_R)$ and denote $p_0 = \sum_{j=1}^D p_j$. Then,
\begin{align*}
& \pi(\Phi | \bB, \bomega) \,\propto\, \pi(\Phi) \int_0^\infty \pi(\bB | \bomega, \Phi, \tau) \pi(\tau) d\tau \\
& \,\propto\, \Big[\prod_{r = 1}^R \phi_r^{\alpha_r - 1}\Big] \int_0^\infty \prod_{r=1}^R \Big[(\tau \phi_r)^{-p_0 / 2} \exp\Big(-\frac{1}{\tau \phi_r} \sum_{j = 1}^{d} || \bbeta_{jr} ||^2 / (2\omega_{jr}) \Big) \Big] \tau^{a_\tau - 1} \exp(- b_\tau \tau) d\tau \\
& \,\propto\, \Big[\prod_{r = 1}^R \phi^{\alpha_r - \tfrac{p_0}{2} - 1}\Big] \int_0^\infty \tau^{a_\tau - R \tfrac{p_0}{2} - 1} \prod_{r=1}^R \exp\Big(-\frac{C_r}{\tau \phi_r} - b_\tau(\tau \phi_r)\Big) \: d\tau
\end{align*}
with $C_r = \sum_{j=1}^d ||\bbeta_{jr}||^2 / (2 \omega_{jr})$. When $a_\tau = \sum_{r=1}^R \alpha_r$, this contains the kernel of a generalized inverse Gaussian (gIG) distribution for $(\tau\phi_r)$.   Recall: $X \sim f_X(x) = \mathrm{giG}(p, a, b) \,\propto\, x^{p-1} \exp(-(a x + b / x) / 2)$.
Following Lemma \ref{lemma:normmeasure}, for independent random variable $T_r \sim f_r$ on $(0, \infty)$, the  joint density of $\{\phi_r = T_r / \sum_{\tilde r} T_{\tilde r} : r = 1, \dots, R\}$ has support on $\mathcal{S}^{R-1}$. In particular,
\[
f(\phi_1, \dots, \phi_{R-1}) = \int_0^\infty t^{R-1} \prod_{r = 1}^R f_r(\phi_r t) ~dt, \quad \phi_R = 1 - \sum_{r < R} \phi_r.
\]
Substituting $f_r(x) \,\propto\, x^{-\delta_r} \exp(-C_r / x) \exp(-b_\tau x)$ in the above expression yields
\begin{align*}
f(\phi_1, \dots, \phi_{R-1}) &\,\propto\, \int_0^\infty \tau^{R-1} \prod_{r = 1}^R (\phi_r \tau)^{-\delta_r} \exp\Big(-\frac{C_r}{(\phi_r \tau)} - b_\tau (\phi_r \tau) \Big) d\tau \\
&= \Big[\prod_{r = 1}^R \phi^{-\delta_r}\Big] \int_0^\infty \tau^{R-\sum_r \delta_r -1} \prod_{r=1}^R\exp\Big(-\frac{C_r}{(\phi_r \tau)} - b_\tau (\phi_r \tau) \Big) d\tau.
\end{align*}
Matching exponents between this expression and the preceding one implies (1) $a_\tau - R(p_0 / 2) - 1 = R - \sum_r \delta_r - 1$, and (2) $\delta_r = 1 + p_0/2 - \alpha_r$. Then,
\begin{align*}
a_\tau = R(1 + p_0/2) - \sum_r \delta_r  = R(1 + p_0/2) - (R + Rp_0 / 2 - \sum_r \alpha_r) = \sum_r \alpha_r
\end{align*}
as previously noted. Hence, draws from $[\Phi | \alpha, \bB, \bW]$ are obtained by sampling $T_r \sim f_r = \mathrm{giG}(\alpha_r - p_0/2, 2b_\tau, 2C_r)$ independently for $r = 1, \dots, R$, and renormalizing.


\subsection*{Proof of lemma~\ref{lem:dl_gdp_var}} \label{sec:prior_moment_pfs}
\begin{proof}
Using priors defined in \eqref{eq:M1}, one has $C_\lambda = \Exp_\lambda (1 / \lambda^2) = \frac{b_\lambda^{2}}{(a_\lambda -1)(a_\lambda - 2)}$ for any $a_\lambda > 2$. In addition, the following inequalities are useful to bound the latter quantity:
\begin{itemize}
\item If $\alpha_1 = c / R, \, c \in \mathbb{N}_{+}$, $\Gamma(\alpha_0 + D) / \Gamma(\alpha_0) = \alpha_0 (\alpha_0 + 1) \cdots (\alpha_0 + D - 1).$ Using the fact that $\log(x + 1) \le x, \, x \ge 0$, one has $\log(\alpha_0) + \cdots + \log(\alpha_0 + D - 1) \le \alpha_0 D -1 + \sum_{k = {1 \lor D-2}}^{D-2} k$. Then $\alpha_0^D \le \Gamma(\alpha_0 + D) / \Gamma(\alpha_0) \le A_\tau \exp(\alpha_0 D)$ where $A_\tau = \exp( -1 + \sum_{k = {1 \lor D-2}}^{D-2} k) = \exp\big((D^2 - 3D) / 2\big)$, $D \ge 2$.
\item Trivially, $||\Phi||_D^D \le 1$; in addition, by H{\"o}lder's inequality, for any $x \in \Re^k$ and $0 < r < p$, one has $||x||_p \ge k^{-\left(\frac{1}{r} - \frac{1}{p}\right)} || x||_r$. In our setting, $D \ge 2$. Taking $r = 1$ in the latter yields $||\Phi||_D^D \ge R^{-(D-1)}.$
\end{itemize}
Recall $\alpha_0 = \sum_{r=1}^R \alpha_r = \alpha_1 R$. This leads to the lower and upper bounds for the prior voxel-level variance:
\begin{align*}
& \var(B_{i_1, \dots, i_D}) \ge (2 C_\lambda)^D \, (\alpha_1 R)^D R^{-(D-1)} / b_\tau^D\ = (2C_\lambda)^D \, \alpha_1^D R / b_\tau^D\\
& \var(B_{i_1, \dots, i_D}) \le A_\tau (2 C_\lambda)^D \, \exp(\alpha_1 \, R D) / b_\tau^D.
\end{align*}
\end{proof}

\subsection*{Consistency proofs}

The proof of Theorem \ref{theorem:main} relies in part on the existence of exponentially consistent tests.
\paragraph{\bf Definition}
An exponentially consistent sequence of test functions $\Phi_n=I(\by_n\in\mathcal{C}_n)$ for testing $H_0:\bB_n=\bB_n^0$ vs. $H_1:\bB_n\neq \bB_n^0$ satisfies
\begin{equation*}
\Exp_{\bB_n^0}(\Phi_n)\leq c_1\exp(-b_1n), \qquad \sup\limits_{\bB_n\in\mathcal{B}_n^c}\Exp_{\bB_n}(1-\Phi_n)\leq c_2\exp(-b_2n)
\end{equation*}
for some $c_1, c_2, b_1, b_2>0$.

\begin{lemma}\label{lemma:numconc}
Suppose $R\sum_{j=1}^D p_{j,n}=o(n)$, then there exist an exponentially consistent sequence of tests $\Phi_n$ for testing $H_0:\bB_n=\bB_n^0$ vs. $H_1:\bB_n\neq \bB_n^0$.
\end{lemma}
\begin{proof}
We begin by stating that $-2(l(\bB_n^0)-l(\hat{\bB}_n))\sim\chi_{R\sum_{j=1}^D p_{j,n}}^2$ under $\bB_n^0$.  We choose the critical region of the test 
$\Phi_n$ as
 $\mathcal{C}_n=\left\{\bB_n: \left|\frac{2}{n}(l(\bB_n^0)-l(\hat{\bB}_n))\right|>\epsilon/4\right\}$.  Note that
\begin{align*}
&\Exp_{\bB_n^0}(\Phi_n)=P_{\bB_n^0}\Big(-\frac{2}{n}(l(\bB_n^0)-l(\hat{\bB}_n))>\epsilon/4\Big) =P_{\bB_n^0}\Big(\chi_{R\sum_{j=1}^D p_{j,n}}^2>n\epsilon/4\Big)\\
&\leq \exp\Big(-\frac{n\epsilon}{16}\Big),\:\mbox{for large n},
\end{align*}
where the last line follows by simplifying Laurent and Massart (2000) and using $P(\chi_p^2>x)<\exp(-x/4)$ if $x\geq 8p$.

Now we will use the fact that
\begin{align*}
\frac{2}{n}(l(\bB_n^0)-l(\hat{\bB}_n))&=\frac{1}{n}\sum_{j=1}^D\left(y_i-\langle \bX_i,\bB_n\rangle\right)^2-\frac{1}{n}\sum_{j=1}^D\left(y_i-\langle \bX_i,\bB_n^0\rangle\right)^2\\
&=\frac{1}{n}\sum_{i=1}^n \left(\langle \bX_i,\bB_n-\bB_n^0\rangle\right)^2+\frac{2}{n}\sum_{i=1}^n(y_i-\langle \bX_i,\bB_n\rangle)\langle\bX_i, \bB_n^0-\bB_n\rangle\\
&=\frac{1}{n}\sum_{i=1}^n KL_i +\frac{2}{n}\sum_{i=1}^n(y_i-\langle \bX_i,\bB_n\rangle)\langle\bX_i, \bB_n^0-\bB_n\rangle.
\end{align*}
Note that, under $\bB_n$,
\begin{align*}
\frac{2}{n}\sum_{i=1}^n(y_i-\langle \bX_i,\bB_n\rangle)\langle\bX_i, \bB_n^0-\bB_n\rangle\sim N(0,\frac{4}{n^2}\sum_{i=1}^{n}KL_i).
\end{align*}
Thus,
\begin{align*}
& \sup\limits_{\bB_n\in \mathcal{B}_n^c}\Exp_{\bB_n}(1-\Phi_n)=\sup\limits_{\bB_n\in \mathcal{B}_n^c}
P_{\bB_n}\left(|\frac{2}{n}(l(\bB_n^0)-l(\hat{\bB}_n))|\leq\epsilon/4\right)\\
&=\sup\limits_{\bB_n\in \mathcal{B}_n^c}
P_{\bB_n}\left(\left|\left|\frac{2}{n}(l(\bB_n)-l(\bB_n^0))\right|-\left|\frac{2}{n}(l(\bB_n)-l(\hat{\bB}_n))\right|\right|\leq\epsilon/4\right)\\
&\leq \sup\limits_{\bB_n\in \mathcal{B}_n^c}
P_{\bB_n}\left(\left|\frac{2}{n}(l(\bB_n)-l(\bB_n^0))\right|-\epsilon/4\leq\left|\frac{2}{n}(l(\bB_n)-l(\hat{\bB}_n))\right|\right)\\
&\leq \sup\limits_{\bB_n\in \mathcal{B}_n^c}
P_{\bB_n}\left(\left|\frac{1}{n}\sum_{i=1}^n KL_i+\sqrt{\frac{\sum_{i=1}^{n}KL_i}{n}}\frac{Z}{\sqrt{n}}\right|-\epsilon/4\leq\left|\frac{2}{n}(l(\bB_n)-l(\hat{\bB}_n))\right|\right)\\
&\leq \sup\limits_{\bB_n\in \mathcal{B}_n^c}
P_{\bB_n}\left(\left|\frac{1}{n}\sum_{i=1}^n KL_i+\sqrt{\frac{\sum_{i=1}^{n}KL_i}{n}}\frac{Z}{\sqrt{n}}\right|-\epsilon/4\leq\left|\frac{2}{n}(l(\bB_n)-l(\hat{\bB}_n))\right|\right),
\end{align*}
Where $Z\sim N(0,1)$. Let $\mathcal{T}_n=\left\{\left|\sqrt{\frac{\sum_{i=1}^{n}KL_i}{n}}\frac{Z}{\sqrt{n}}\right|\leq \frac{1}{2n}\sum_{i=1}^n KL_i\right\}$. Using this fact we have
{\small\begin{align*}
& \sup\limits_{\bB_n\in \mathcal{B}_n^c}\Exp_{\bB_n}(1-\Phi_n)\\
&\leq \sup\limits_{\bB_n\in \mathcal{B}_n^c}
P_{\bB_n}\left(\left\{\left|\frac{1}{n}\sum_{i=1}^n KL_i+\sqrt{\frac{\sum_{i=1}^{n}KL_i}{n}}\frac{Z}{\sqrt{n}}\right|-\epsilon/4\leq\left|\frac{2}{n}(l(\bB_n)-l(\hat{\bB}_n))\right|\right\}\cap\mathcal{T}_n\right)
+ \sup\limits_{\bB_n\in \mathcal{B}_n^c}
P_{\bB_n}\left(\mathcal{T}_n\right)\\
&\leq \sup\limits_{\bB_n\in \mathcal{B}_n^c}
P_{\bB_n}\left(\frac{1}{2n}\sum_{i=1}^n KL_i-\epsilon/4\leq\left|\frac{2}{n}(l(\bB_n)-l(\hat{\bB}_n))\right|\right)
+ \sup\limits_{\bB_n\in \mathcal{B}_n^c}
P_{\bB_n}\left(\left|\frac{Z}{\sqrt{n}}\right|\geq \frac{1}{2}\sqrt{\frac{1}{n}\sum_{i=1}^n KL_i}\right)\\
&\leq P_{\bB_n}\left(\frac{3\epsilon}{4}\leq\left|\frac{2}{n}(l(\bB_n)-l(\hat{\bB}_n))\right|\right)+
P_{\bB_n}\left(\left|Z\right|\geq \frac{1}{2}\sqrt{n\epsilon}\right)\\
&\leq P_{\bB_n}\left(\frac{3n\epsilon}{4}\leq\chi_{R\sum_{j=1}^{D}p_{j,n}}^2\right)+
P_{\bB_n}\left(\chi_1^2\geq \frac{n\epsilon}{4}\right)\leq \exp\left(-\frac{3n\epsilon}{16}\right)+\exp\left(-\frac{n\epsilon}{16}\right)
\leq 2\exp\left(-\frac{n\epsilon}{16}\right),
\end{align*}}
where the last line requires an application of Laurent and Massart (2000).
\end{proof}
\noindent
{\bf Theorem \ref{theorem:main}}
\begin{proof}
Under Lemma \ref{lemma:numconc} one has
\begin{align*}
\Pi_n(\mathcal{B}_n^c) =\frac{\int_{\mathcal{B}_n^c}f(\by_n|\bB_n)\pi_n(\bF_n)}{\int f(\by_n|\bB_n)\pi_n(\bF_n)}
=\frac{\int_{\mathcal{B}_n^c}\frac{f(\by_n|\bB_n)}{f(\by_n|\bB_n^0)}\pi_n(\bF_n)}{\int \frac{f(\by_n|\bB_n)}{f(\by_n|\bB_n^0)}\pi_n(\bF_n)}
= \frac{N}{D}\leq \Phi_n+(1-\Phi_n)\frac{N}{D}.
\end{align*}
Note that we have
\begin{align*}
P_{\bB_n^0}\left(\Phi_n>\exp(-b_1n/2)\right)\leq \Exp_{\bB_n^0}\left(\Phi_n\right)\exp(b_1n/2)\leq c_1\exp(-b_1n/2).
\end{align*}
Therefore $\sum_{n=1}^{\infty}P_{\bB_n^0}\left(\Phi_n>\exp(-b_1n/2)\right)<\infty$. Using Borel-Cantelli lemma\\ $P_{\bB_n^0}\left(\Phi_n>\exp(-b_1n/2) i.o.\right)=0$. It follows that
\begin{align}\label{eq:show1}
\Phi_n \rightarrow 0 \quad a.s.
\end{align}
In addition, we have
\begin{align*}
\Exp_{\bB_n^0}((1-\Phi_n)N) &=\int (1-\Phi_n)\int_{\mathcal{B}_n^c}\frac{f(\by_n|\bB_n)}{f(\by_n|\bB_n^0)}\pi_n(\bF_n)f(\by_n|\bB_n^0)\\
&=\int_{\mathcal{B}_n^c}\int (1-\Phi_n)f(\by_n|\bB_n)\pi_n(\bF_n)\\
&\leq \sup\limits_{\bB_n\in\mathcal{B}_n^c}\Exp_{\bB_n}(1-\Phi_n)\leq c_2\exp(-b_2n).
\end{align*}
Using a similar technique as above,  $P_{\bB_n^0}\left((1-\Phi_n)N\exp(nb_2/2)>\exp(-nb_2/4) i.o.\right)=0$ so
\begin{align}
\exp(bn)(1-\Phi_n)N \rightarrow 0 \quad a.s.\label{eq:show2}.
\end{align}
By Lemma~\ref{lemma:numconc} and (\ref{eq:show1})-(\ref{eq:show2}) it is enough to show that
 $M=\exp(\tilde{b}n)\int\frac{f(\by_n|\bB_n)}{f(\by_n|\bB_n^0)}\pi_n(\bF_n)\rightarrow\infty$ for some $\tilde{b}\leq b=\frac{\epsilon}{32}$. We choose $\tilde{b}=b$. Consider the set
$\mathcal{H}_n=\left\{\bB_n:\frac{1}{n}\log\left[\frac{f(\by_n|\bB_n)}{f(\by_n|\bB_n^0)}\right]<\eta\right\}$, for some $\eta$ which is chosen later.
\begin{align*}
M &\geq \exp(\tilde{b}n)\int\limits_{\mathcal{H}_n}\exp\left(-n\frac{1}{n}\frac{f(\by_n|\bB_n)}{f(\by_n|\bB_n^0)}\right)\pi_n(\bF_n)\\
&\geq \exp((\tilde{b}-\eta)n)\pi_n(\mathcal{H}_n).
\end{align*}
Note that
\begin{align*}
&\frac{1}{n}\log\left[\frac{f(\by_n|\bB_n)}{f(\by_n|\bB_n^0)}\right]\\
&=\frac{1}{n}\left[-\frac{1}{2}\sum_{i=1}^{n}(y_i-\langle \bX_i,\bB_n\rangle)^2+\frac{1}{2}\sum_{i=1}^{n}(y_i-\langle \bX_i,\bB_n^0\rangle)^2\right].
\end{align*}
Let $\by_n=(y_1,...,y_n)'$, $\bH_n=\left(\langle \bX_1,\bB_n\rangle,\ldots,\langle \bX_n,\bB_n\rangle\right)$
and $\bH_n^0=\left(\langle \bX_1,\bB_n^0\rangle,\ldots,\langle \bX_n,\bB_n^0\rangle\right)$. Then
{\small\begin{align*}
&\pi_n\Big(\bB_n:\frac{1}{n}\left[-||\by_n-\bH_n^0||^2+||\by_n-\bH_n||^2\right]<2\eta\Big)\\
&\geq\pi_n\Big(\bB_n:\frac{1}{n}\left|2||\by_n-\bH_n^0||\left(||\by_n-\bH_n||-||\by_n-\bH_n^0||\right)+\left(
||\by_n-\bH_n||-||\by_n-\bH_n^0||\right)^2\right|<2\eta\Big)\\
&\geq \pi_n\Big(\bB_n:\frac{1}{n}\left|2||\by_n-\bH_n^0|| ||\bH_n^0-\bH_n||+||\bH_n-\bH_n^0||^2\right|<2\eta\Big)\\
&\geq \pi_n\Big(\bB_n:\frac{1}{n}||\bH_n^0-\bH_n||<\frac{2\eta}{3\zeta_n},||\by_n-\bH_n^0||^2<\zeta_n^2\Big)\\
&\geq \pi_n\big(\mathcal{A}_{1n}\cap\mathcal{A}_{2n}\big)
\end{align*}}
where $\mathcal{A}_{1n}=\left\{\frac{1}{n}||\bH_n-\bH_n^0||<\frac{2\eta}{3\zeta_n}\right\}$,
$\mathcal{A}_{2n}=\left\{||\by_n-\bH_n^0||^2<\zeta_n^2\right\}$.

We will show that $P_{\bB_n^0}(\mathcal{A}_{2n})=1$ for all large $n$. Assume $\zeta_n=n^{(1+\rho_3)/2},\:\rho_3>0$ so that $\zeta_n^2>8n$ for all large $n$.
Then,
\begin{align*}
P_{\bB_n^0}(\mathcal{A}_{2n'})=P_{\bB_n^0}(\chi_n^2>\zeta_n^2)\leq \exp(-\zeta_n^2/2).
\end{align*}
Therefore, using Borel-Cantelli lemma $P_{\bB_n^0}(\mathcal{A}_{2n}'\: i.o.)=0$. Hence $P_{\bB_n^0}(\mathcal{A}_{2n})=1$ for all large $n$.
It is enough to bound $\pi_n(\mathcal{A}_{1n})$.  Let $M_n=\frac{1}{n}\sqrt{\sum_{i=1}^{n}||\bX_i||_2^2}$. Now use the fact that $\frac{1}{n}||\bH_n-\bH_n^0||=\frac{1}{n}\sqrt{\sum_{i=1}^n(\langle \bX_i,\bB_n-\bB_n^0\rangle)^2} \leq \left(\frac{1}{n}\sqrt{\sum_{i=1}^{n}||\bX_i||_2^2}\right)||\bB_n-\bB_n^0||_2$ to conclude
\begin{align}\label{eq:firstineq}
\left\{||\bB_n-\bB_n^0||_2<\frac{2\eta}{3M_n\zeta_n}\right\}\subseteq\mathcal{A}_{1n}.
\end{align}
By (\ref{priorprop}) one has $\pi_n(\mathcal{A}_{1n})\geq\pi_n\Big(||\bB_n-\bB_n^0||_2<\frac{2\eta}{3M_n\zeta_n}\Big)\geq \exp(-dn)$ and hence $M\geq \exp\big((\tilde{b}-\eta-d)n\big)\rightarrow\infty$ as $n\rightarrow\infty$ proving the result.
\end{proof}

\noindent
{\bf Theorem ~\ref{tensorDLGDP}}
\begin{proof}
Define $g:\mathbb{R}\rightarrow\mathbb{R}$ s.t.
\begin{align*}
g(\kappa)=R\kappa^D+\kappa^{D-1}\sum_{j=1}^{D}\sum_{r=1}^{R}||\bbeta_{j,n}^{0(r)}||_2+\cdots+\kappa\sum_{j=1}^{D}\sum_{r=1}^{R}\prod\limits_{l\neq j}||\bbeta_{l,n}^{0(r)}||_2.
\end{align*}
Let $\kappa_n>0$ be s.t. $g(\kappa_n)=\frac{2\eta}{3M_n\zeta_n}$. Note that by Decarte's rule of sign, the equation $g(\kappa)-\frac{2\eta}{3M_n\zeta_n}=0$ has a unique positive root. Further
\begin{align}
\frac{1}{\kappa_n}<1+\max\limits_{i=1,...,D}\left|\frac{3\sum_{j_1\neq\cdots\neq j_i}\sum_{r=1}^{R}\prod\limits_{l=1}^i||\bbeta_{j_l,n}^{0(r)}||_2}{2\eta/M_n\zeta_n}\right|\label{eq:root1}\\
\kappa_n<1+\max\left\{\frac{2\eta}{3M_n\zeta_nR},\max\limits_{i=1,...,D}\left|\frac{\sum_{j_1\neq\cdots\neq j_i}\sum_{r=1}^{R}\prod\limits_{l=1}^i||\bbeta_{j_l,n}^{0(r)}||_2}{R}\right|\right\}\label{eq:root2}
\end{align}
by Lemma \ref{lem:polynomroots}.

Using Lemma \ref{lem:tensordiff} it is easy to see that
\begin{align}\label{eq:secondineq}
\left\{||\bbeta_{j,n}^{(r)}-\bbeta_{j,n}^{0(r)}||_2\leq \kappa_n,\:j=1,...,D;\:r=1,...,R\right\}\subseteq\left\{||\bB_n-\bB_n^0||_2<\frac{2\eta}{3M_n\zeta_n}\right\}.
\end{align}
Using (\ref{eq:firstineq}), $\pi_n\left(||\bB_n-\bB_n^0||_2<\frac{2\eta}{3M_n\zeta_n}\right)\geq \pi_n\left(\left\{||\bbeta_{j,n}^{(r)}-\bbeta_{j,n}^{0(r)}||_2\leq \kappa_n,\:j=1,...,D;\:r=1,...,R\right\}\right)$. Note that
\begin{align*}
&\pi_n\left(\left\{||\bbeta_{j,n}^{(r)}-\bbeta_{j,n}^{0(r)}||_2\leq \kappa_n,\:j=1,...,D;\:r=1,...,R\right\}|\{w_{jr,l}\}_{l=1}^{p_{j,n}},\{\lambda_{jr}\}_{j,r=1}^{D,R-1},\{\phi_r\}_{r=1}^{R-1},\tau\right)\\
&\geq \left[\prod\limits_{j=1}^{D}\prod\limits_{r=1}^{R}\pi_n\left(||\bbeta_{j,n}^{(r)}-\bbeta_{j,n}^{0(r)}||_2\leq \kappa_n|\{w_{jr,l}\}_{l=1}^{p_{j,n}},\{\lambda_{jr}\}_{j,r=1}^{D,R-1},\{\phi_r\}_{r=1}^{R-1},\tau\right)\right].
\end{align*}
Therefore, it is enough to bound $\pi_n(||\bbeta_{j,n}^{(r)}-\bbeta_{j,n}^{0(r)}||\leq \kappa_n, j=1,...,D; r=1,...,R)$. For $j=1,...,D$, $r=1,...,R$,
\begin{align*}
&\pi_n(||\bbeta_{j,n}^{(r)}-\bbeta_{j,n}^{0(r)}||\leq \kappa_n|\{w_{jr,l}\}_{l=1}^{p_{j,n}},\lambda_{jr},\{\phi_r\}_{r=1}^{R-1},\tau)\\
&\quad\geq\prod_{l=1}^{p_{j,n}} \pi_n\left(|\beta_{j,n,l}^{(r)}-\beta_{j,n,l}^{0(r)}|\leq \frac{\kappa_n}{\sqrt{p_{j,n}}}|\{w_{jr,l}\}_{l=1}^{p_{j,n}},\lambda_{jr},\{\phi_r\}_{r=1}^{R-1},\tau\right)\nonumber\\
&\geq \prod_{l=1}^{p_{j,n}}\left\{\left(\frac{2\kappa_n}{\sqrt{2p_{j,n}\pi w_{jr,l}\phi_r\tau}}\right)\exp\left(-\frac{|\beta_{j,n,l}^{0(r)}|^2+\kappa_n^2/p_{j,n}}{w_{jr,l}\phi_r\tau}\right)\right\},
\end{align*}
where the last step follows from the fact that $\int_{a}^{b}e^{-x^2/2}dx\geq e^{-(a^2+b^2)/2}(b-a)$.
Thus,
\begin{align}\label{eq:GDPnew}
&\pi_n(||\bbeta_{j,n}^{(r)}-\bbeta_{j,n}^{0(r)}||\leq \kappa_n|\lambda_{jr},\{\phi_r\}_{r=1}^{R-1},\tau)\nonumber\\
&=\Exp\left[\pi_n(||\bbeta_{j,n}^{(r)}-\bbeta_{j,n}^{0(r)}||\leq \kappa_n|\{w_{jr,l}\}_{l=1}^{p_{j,n}},\lambda_{jr},\{\phi_r\}_{r=1}^{R-1},\tau)\right]\nonumber\\
&\geq \left(\frac{2\kappa_n}{\sqrt{2p_{j,n}\pi\phi_r\tau}}\right)^{p_{j,n}}
\prod_{l=1}^{p_{j,n}}E\left\{\frac{1}{\sqrt{w_{jr,l}}}\exp\left(-\frac{|\beta_{j,n,l}^{0(r)}|^2+\kappa_n^2/p_{j,n}}{w_{jr,l}\phi_r\tau}\right)\right\}\nonumber\\
&\geq \left(\frac{2\kappa_n\lambda_{jr}^2}{2\sqrt{2p_{j,n}\pi\phi_r\tau}}\right)^{p_{j,n}}
\prod_{l=1}^{p_{j,n}}\int_{w_{jr,l}}\left\{\frac{1}{\sqrt{w_{jr,l}}}\exp\left(-\frac{|\beta_{j,n,l}^{0(r)}|^2+\kappa_n^2/p_{j,n}}{w_{jr,l}\phi_r\tau}-\frac{\lambda_{jr}^2w_{jr,l}}{2}\right)dw_{jr,l}\right\}.
\end{align}
Use the change of variable $\frac{1}{w_{jr,l}}=z_{jr,l}$ and the normalizing constant from the inverse Gaussian density to deduce
\begin{align*}
&\int_{w_{jr,l}}\left\{\frac{1}{\sqrt{w_{jr,l}}}\exp\left(-\frac{|\beta_{j,n,l}^{0(r)}|^2+\kappa_n^2/p_{j,n}}{w_{jr,l}\phi_r\tau}-\frac{\lambda_{jr}^2w_{jr,l}}{2}\right)dw_{jr,l}\right\}\\
&\qquad=\int_{z_{jr,l}}\left\{\frac{1}{\sqrt{z_{jr,l}^3}}\exp\left(-\frac{(|\beta_{j,n,l}^{0(r)}|^2+\kappa_n^2/p_{j,n})}{\phi_r\tau}z_{jr,l}-\frac{\lambda_{jr}^2}{2z_{jr,l}}\right)dz_{jr,l}\right\}\\
&\qquad=\sqrt{\left(\frac{2\pi}{\lambda_{jr}^2}\right)}\exp\left(-\lambda_{jr}\frac{\sqrt{2\left(|\beta_{j,n,l}^{0(r)}|^2+\kappa_n^2/p_{j,n}\right)}}{\sqrt{\phi_r\tau}}\right).
\end{align*}
(\ref{eq:GDPnew}) can be written as
\begin{align*}
&\pi_n(||\bbeta_{j,n}^{(r)}-\bbeta_{j,n}^{0(r)}||\leq \kappa_n|\lambda_{jr},\{\phi_r\}_{r=1}^{R-1},\tau)\\
&\geq \left(\frac{2\kappa_n\lambda_{jr}^2}{2\sqrt{2p_{j,n}\pi\phi_r\tau}}\right)^{p_{j,n}}
\prod_{l=1}^{p_{j,n}}\left[\sqrt{\left(\frac{2\pi}{\lambda_{jr}^2}\right)}\exp\left(-\lambda_{jr}\frac{\sqrt{2\left(|\beta_{j,n,l}^{0(r)}|^2+\kappa_n^2/p_{j,n}\right)}}{\sqrt{\phi_r\tau}}\right)\right]\\
&=\left(\frac{2\kappa_n\lambda_{jr}}{2\sqrt{p_{j,n}\phi_r\tau}}\right)^{p_{j,n}}\exp\left(-\lambda_{jr}\frac{\sum_{l=1}^{p_{j,n}}\sqrt{2\left(|\beta_{j,n,l}^{0(r)}|^2+\kappa_n^2/p_{j,n}\right)}}{\sqrt{\phi_r\tau}}\right).
\end{align*}
Therefore,
\begin{align*}
&\pi_n(||\bbeta_{j,n}^{(r)}-\bbeta_{j,n}^{0(r)}||\leq \kappa_n|\{\phi_r\}_{r=1}^{R-1},\tau)\\
&\geq \left(\frac{2\kappa_n}{2\sqrt{p_{j,n}\phi_r\tau}}\right)^{p_{j,n}}\frac{b_{\lambda,r}^{a_{\lambda,r}}}{\Gamma(a_{\lambda,r})}\int_{\lambda_{jr}}\lambda_{jr}^{p_{j,n}+a_{\lambda,r}-1}
\exp\left(-\lambda_{jr}\left[\frac{\sum_{l=1}^{p_{j,n}}\sqrt{2\left(|\beta_{j,n,l}^{0(r)}|^2+\kappa_n^2/p_{j,n}\right)}}{\sqrt{\phi_r\tau}}+b_{\lambda,r}\right]\right)d\lambda_{jr}\\
&=\left(\frac{2\kappa_n}{2\sqrt{p_{j,n}\phi_r\tau}}\right)^{p_{j,n}}\frac{b_{\lambda,r}^{a_{\lambda,r}}}{\Gamma(a_{\lambda,r})}
\frac{\Gamma(p_{j,n}+a_{\lambda,r})}{\left[\frac{\sum_{l=1}^{p_{j,n}}\sqrt{2\left(|\beta_{j,n,l}^{0(r)}|^2+\kappa_n^2/p_{j,n}\right)}}{\sqrt{\phi_r\tau}}+b_{\lambda,r}\right]^{p_{j,n}+a_{\lambda,r}}}\\
&=\left(\frac{2\kappa_n}{2b_{\lambda,r}\sqrt{p_{j,n}\phi_r\tau}}\right)^{p_{j,n}}\frac{1}{\Gamma(a_{\lambda,r})}
\frac{\Gamma(p_{j,n}+a_{\lambda,r})}{\left[\frac{\sum_{l=1}^{p_{j,n}}\sqrt{2\left(|\beta_{j,n,l}^{0(r)}|^2+\kappa_n^2/p_{j,n}\right)}}{b_{\lambda,r}\sqrt{\phi_r\tau}}+1\right]^{p_{j,n}+a_{\lambda,r}}}.
\end{align*}
The final expression as in the above yields
\begin{align*}
&\pi_n(||\bbeta_{j,n}^{(r)}-\bbeta_{j,n}^{0(r)}||\leq \kappa_n, j=1,...,D,r=1,...,R|\{\phi_r\}_{r=1}^{R-1},\tau)\\
&\geq \Exp\left\{\prod_{j=1}^D\prod_{r=1}^R\left[\left(\frac{2\kappa_n}{2b_{\lambda,r}\sqrt{p_{j,n}\phi_r\tau}}\right)^{p_{j,n}}\frac{1}{\Gamma(a_{\lambda,r})}\lambda_{j,r}^{p_{j,n}+a_{\lambda,r}-1}
\frac{\Gamma(p_{j,n}+a_{\lambda,r})}{\left[\frac{\sum_{l=1}^{p_{j,n}}\sqrt{2\left(|\beta_{j,n,l}^{0(r)}|^2+\kappa_n^2/p_{j,n}\right)}}{b_{\lambda,r}\sqrt{\phi_r\tau}}+1\right]^{p_{j,n}+a_{\lambda,r}}}\right]\right\}.
\end{align*}
We will now use the fact that for $\phi_r\leq 1$,
\begin{align*}
\frac{1}{\left[\frac{\sum_{l=1}^{p_{j,n}}\sqrt{2\left(|\beta_{j,n,l}^{0(r)}|^2+\kappa_n^2/p_{j,n}\right)}}{b_{\lambda,r}\sqrt{\phi_r\tau}}+1\right]^{p_{j,n}+a_{\lambda,r}}}
\geq \frac{1}{\left[\frac{\sum_{l=1}^{p_{j,n}}\sqrt{2\left(|\beta_{j,n,l}^{0(r)}|^2+\kappa_n^2/p_{j,n}\right)}}{b_{\lambda,r}\sqrt{\phi_r\tau}}+\frac{1}{\sqrt{\tau\phi_r}}\right]^{p_{j,n}+a_{\lambda,r}}}I_{\tau\in[0,1]}.
\end{align*}
This inequality is critical to provide a lower bound on $\pi_n(||\bbeta_{j,n}^{(r)}-\bbeta_{j,n}^{0(r)}||\leq \kappa_n, j=1,...,D,r=1,...,R)$ as following
\begin{align*}
&\pi_n(||\bbeta_{j,n}^{(r)}-\bbeta_{j,n}^{0(r)}||\leq \kappa_n, j=1,...,D,r=1,...,R)\\
&\geq \frac{\lambda_2^{\lambda_1}\Gamma(Ra)}{\Gamma(\lambda_1)\Gamma(a)^R}\prod_{j=1}^D\prod_{r=1}^R\left[\left(\frac{\kappa_n}{\sqrt{p_{j,n}}b_{\lambda,r}}\right)^{p_{j,n}}
\frac{\Gamma(p_{j,n}+a_{\lambda,r})}{\Gamma(a_{\lambda,r})}\right]
\int_{\tau}\tau^{\lambda_1-R\sum_{j=1}^D \frac{p_{j,n}}{2}-1}\exp(-\lambda_2\tau)\\
&\int_{\bphi\in\mathcal{S}^{R-1}}\frac{\prod_{r=1}^{R}\phi_{r}^{a-1}}{\prod_{r=1}^{R}\phi_{r}^{\sum_{j=1}^D\frac{p_{j,n}}{2}}}
 \prod_{j=1}^D\prod_{r=1}^R\frac{1}{\left[\frac{\sum_{l=1}^{p_{j,n}}
\sqrt{2\left(|\beta_{j,n,l}^{0(r)}|^2+\kappa_n^2/p_{j,n}\right)}}{b_{\lambda,r}\sqrt{\phi_r\tau}}+1\right]^{p_{j,n}+a_{\lambda,r}}}d\bphi d\tau\\
&\geq \frac{\lambda_2^{\lambda_1}\Gamma(Ra)}{\Gamma(\lambda_1)\Gamma(a)^R}\prod_{j=1}^D\prod_{r=1}^R\left[\left(\frac{\kappa_n}{\sqrt{p_{j,n}}b_{\lambda,r}}\right)^{p_{j,n}}
\frac{\Gamma(p_{j,n}+a_{\lambda,r})}{\Gamma(a_{\lambda,r})}\right]\prod_{j=1}^D\prod_{r=1}^R\frac{1}{\left[\frac{\sum_{l=1}^{p_{j,n}}
\sqrt{2\left(|\beta_{j,n,l}^{0(r)}|^2+\kappa_n^2/p_{j,n}\right)}}{b_{\lambda,r}}+1\right]^{p_{j,n}+a_{\lambda,r}}}\\
&\left(\int_{\tau=0}^1\tau^{\lambda_1+\sum_{r=1}^R a_{\lambda,r}\frac{D}{2}-1}\exp(-\tau\lambda_2)d\tau\right)\int_{\bphi\in\mathcal{S}^{R-1}}\prod_{r=1}^{R}\phi_{r}^{a+a_{\lambda,r}\frac{D}{2}-1}d\bphi\\
&= \frac{\lambda_2^{\lambda_1}\Gamma(Ra)}{\Gamma(\lambda_1)\Gamma(a)^R}\prod_{j=1}^D\prod_{r=1}^R\left[\left(\frac{\kappa_n}{\sqrt{p_{j,n}}b_{\lambda,r}}\right)^{p_{j,n}}
\frac{\Gamma(p_{j,n}+a_{\lambda,r})}{\Gamma(a_{\lambda,r})}\right]\prod_{j=1}^D\prod_{r=1}^R\frac{1}{\left[\frac{\sum_{l=1}^{p_{j,n}}
\sqrt{2\left(|\beta_{j,n,l}^{0(r)}|^2+\kappa_n^2/p_{j,n}\right)}}{b_{\lambda,r}}+1\right]^{p_{j,n}+a_{\lambda,r}}}\\
&\times \frac{\exp(-\lambda_2)}{(\lambda_1+\sum_{r=1}^R a_{\lambda,r}\frac{D}{2})}\frac{\prod_{r=1}^R\left[\Gamma(a+a_{\lambda,r}\frac{D}{2})\right]}{\Gamma(Ra+\frac{D}{2}\sum_{r=1}^R a_{\lambda,r})}.
\end{align*}
Denote $C_6=\frac{\lambda_2^{\lambda_1}\Gamma(Ra)}{\Gamma(\lambda_1)\left[\Gamma(a)\right]^R}\frac{\exp(-\lambda_2)}{\left(\lambda_1+\sum_{r=1}^R a_{\lambda,r}\frac{D}{2}\right)}\frac{\prod_{r=1}^R\left[\Gamma(a+a_{\lambda,r}\frac{D}{2})\right]}{\Gamma(Ra+\sum_{r=1}^{R}a_{\lambda,r}\frac{D}{2})}$
. Then the above expression gives us
\begin{align}\label{eq:loglikGDP}
&-\log\left(||\bB_n-\bB_n^0||_2<\frac{2\eta}{3M_n\zeta_n}\right)\nonumber\\
&\leq -\log(C_6)
+\sum_{j=1}^{D}\sum_{r=1}^{R}p_{j,n}\left[-\log(\kappa_n)+\frac{1}{2}\log(p_{j,n})+\log(b_{\lambda,r})+\log(\Gamma(a_{\lambda,r})\right]\nonumber\\
&\qquad-\sum_{r=1}^R\sum_{j=1}^D \log(\Gamma(p_{n,j}+a_{\lambda,r}))
+\sum_{j=1}^D\sum_{r=1}^R(p_{j,n}+a_{\lambda,r})\log\left[\frac{\sum_{l=1}^{p_{j,n}}
\sqrt{2\left(|\beta_{j,n,l}^{0(r)}|^2+\kappa_n^2/p_{j,n}\right)}}{b_{\lambda,r}}+1\right].
\end{align}
Using (\ref{eq:root1}) and assumption (b), it is easy to see that $\frac{1}{\kappa_n}<G_5 n^{\rho_2+\frac{\rho_3+1}{2}}\prod_{j=1}^{D}p_{j,n}$ for a constant $G_5>0$ for all large $n$. Therefore, $\sum_{j=1}^{D}\sum_{r=1}^R p_{j,n}\left[\log\left(\frac{1}{\kappa_n}\right)+\frac{1}{2}\log(p_{j,n})+\log(b_{\lambda,r})+\log(\Gamma(a_{\lambda,r})\right]=o(n)$. Also, $\sum_{r=1}^R\sum_{j=1}^D \log(\Gamma(p_{j,n}+a_{\lambda,r}))]\leq\sum_{j=1}^D(p_{j,n}+a_{\lambda,r})\log(p_{j,n}+a_{\lambda,r})=o(n)$, by assumption (c). Finally,
$\sum_{j=1}^D\sum_{r=1}^R(p_{j,n}+a_{\lambda,r})\log\Big[\frac{\sum_{l=1}^{p_{j,n}}
\sqrt{2\left(|\beta_{j,n,l}^{0(r)}|^2+\kappa_n^2/p_{j,n}\right)}}{b_{\lambda,r}}+1\Big]=o(n)$, by assumptions (b) and (c).
Thus,
$-\log\Big(\pi_n(\bB_n:||\bB_n-\bB_n^0||_2<\frac{2\eta}{3M_n\zeta_n})\Big)<dn$ for all $d>0$, for all large $n$. This proves the result.
\end{proof}

%
%
%
%

\bibliographystyle{jasa}
\bibliography{full_paper_dd_2.bbl}

\clearpage
\section*{Supplemental Materials}\label{sec:supplement}

This supplement contains additional Lemmas relevant to the article, some of which are well known and presented without proof.
\begin{lemma}\label{lem:tensordiff}
Suppose $\bT=T_1\circ\cdots\circ T_D$ and $\bF=F_1\circ\cdots\circ F_D$ are two rank one tensors of same dimension. Then
\begin{align*}
\bT-\bF=(T_1-F_1)\circ\cdots\circ(T_D-F_D)+\sum_{l=1}^{D-1}\sum_{\mathcal{I}_1\cup\mathcal{I}_2=1:D,|\mathcal{I}_1|=l,|\mathcal{I}_2|=D-l}
         \gamma_1\circ\cdots\circ\gamma_D,
\end{align*}
where $\gamma_j=F_j$ if $j\in\mathcal{I}_2$; $=T_j-F_j$ if $j\in\mathcal{I}_1$.
\end{lemma}
\begin{proof}
We will show it by induction. If $D=2$ then,
\begin{align*}
\bT-\bF &=T_1\circ T_2-F_1\circ F_2=(T_1-F_1)\circ T_2+F_1\circ T_2-F_1\circ F_2\\
&=(T_1-F_1)\circ(T_2-F_2)+(T_1-F_1)\circ F_2+F_1\circ(T_2-F_2).
\end{align*}
 Assume the result to hold for $D-1$. For $D$,
\begin{align*}
& T_1\circ\cdots\circ T_D-F_1\circ\cdots\circ F_D\\
&=(T_1-F_1)\circ T_2\circ\cdots\circ T_D+F_1\circ\left[T_2\circ\cdots\circ T_D-F_2\circ\cdots\circ F_D\right]\\
&=(T_1-F_1)\circ[(T_2-F_2)\circ\cdots\circ(T_D-F_D)+F_2\circ\cdots\circ F_D+\\
&\sum_{l=1}^{D-2}\sum\limits_{\mathcal{I}_1\cup\mathcal{I}_2, |\mathcal{I}_1|=l,|\mathcal{I}_2|=D-1-l} \gamma_2\circ\cdots\circ\gamma_D]+\\
&F_1\circ[(T_2-F_2)\circ\cdots\circ(T_D-F_D)+\sum_{l=1}^{D-2}\sum\limits_{\mathcal{I}_1\cup\mathcal{I}_2, |\mathcal{I}_1|=l,|\mathcal{I}_2|=D-1-l} \gamma_2\circ\cdots\circ\gamma_D]\\
&=(T_1-F_1)\circ\cdots\circ(T_D-F_D)+\sum_{l=1}^{D-1}\sum\limits_{\mathcal{I}_1\cup\mathcal{I}_2, |\mathcal{I}_1|=l,|\mathcal{I}_2|=D-l} \gamma_1\circ\cdots\circ\gamma_D].
\end{align*}
Hence proved.
\end{proof}

\begin{lemma}\label{lem:RKHS}
Suppose $\btheta\sim \mathrm{N}(\bzero,\bSigma)$ with $\bSigma$ p.d. and $\btheta_0\in \Re^p$. Let $||\btheta_0||_H=\btheta_0'\bSigma^{-1}\btheta_0$. Then for any
$t>0$
\begin{align*}
\exp(-\frac{||\btheta_0||_H^2}{2})P(||\btheta||_2\leq t/2)\leq P(||\btheta-\btheta_0||_2\leq t)\leq \exp(-\frac{||\btheta_0||_H^2}{2})P(||\btheta||_2\leq t).
\end{align*}
\end{lemma}
\begin{proof}
This is a general version of Anderson's lemma. For more references see Van der Vaart \& Van Zanten (2008).
\end{proof}

\begin{lemma}\label{lem:convexfunc}
If $f_1,\dots,f_d$ are convex functions such that  $f_i>0$ and  $f_i'<0$ for all $i=1,\dots,d$, then $\prod_{i=1}^{d}f_i$ is convex.
\end{lemma}
\noindent\emph{proof of lemma~\ref{lem:convexfunc}}
\begin{proof}
First we prove the result for $d=2$. Note that $(f_1f_2)''=f_1^{''}f_2+f_2^{''}f_1+2f_1'f_2'>0$. So the result holds for $d=2$.
Also $f_1f_2>0$ and $(f_1f_2)'=f_1'f_2+f_2'f_1<0$.

Assume the result to hold for $d-1$, i.e. $\prod_{i=1}^{d-1}f_i$ is convex and $(\prod_{i=1}^{d-1}f_i)'<0$. Then
$(\prod_{i=1}^{d}f_i)^{''}=f_d^{''}\prod_{i=1}^{d-1}f_i+f_d(\prod_{i=1}^{d-1}f_i)^{''}+2f_d'(\prod_{i=1}^{d-1}f_i)'>0$. Hence
$\prod_{i=1}^{d}f_i$ is convex.
\end{proof}

\begin{lemma}\label{lem:funcpos}
Let $g_1(x)=(1+x)^{-\nu/2}$ and $g_2(x)=\frac{c_1 x}{c_2+c_3 x}$, $c_1,c_2,c_3>0$. Then $g_1(x)<\frac{1}{1+x\nu/2}$ and $g_2'(x)>0$ for all $x>0$ and $\nu>2$.
\end{lemma}
\noindent\emph{proof of lemma~\ref{lem:funcpos}}
\begin{proof}
Let $h_1(x)=(1+x)^{\nu/2}-(1+x\nu/2)$, then $h_1'(x)=(1+x)^{\nu/2-1}-\nu/2>0$ for all $x>0, \nu>2$. Further using the fact that $h_1(0)=0$, we conclude
$h_1(x)>0$ for all $x>0,\nu>2$. This implies $g_1(x)<\frac{1}{1+x\nu/2}$. The proof of $g_2 >0$ for all $x>0$ is similar and is omitted.
\end{proof}

\begin{lemma}\label{lem:polynomroots}
Let $x^*$ be a real root of the polynomial $P(x)=a_k x^k+a_{k-1} x^{k-1}+\cdots+a_1 x-a_0$. Then $1/|x^*|<1+\max_{i=1,\dots,k}\left|\frac{a_i}{a_0}\right|$.
\end{lemma}
\noindent\emph{proof of lemma~\ref{lem:polynomroots}}
\begin{proof}
Consider the polynomial $P_1(\zeta)=\zeta^k-\left(\frac{a_1}{a_0}\right)\zeta^{k-1}-\cdots-\left(\frac{a_k}{a_0}\right)$. By making a change of variable with
$\zeta=\frac{1}{x}$, we obtain
\begin{align*}
P_1\left(\frac{1}{x}\right) &=\frac{1}{x^k}-\left(\frac{a_1}{a_0}\right)\frac{1}{x^{k-1}}-\cdots-\left(\frac{a_k}{a_0}\right)\\
&=-\frac{a_k x^k+\cdots+a_1 x-a_0}{a_0 x^k}.
\end{align*}
Note that $P_1\left(\frac{1}{x}\right)=0$ is solved by $x=x^*$. Therefore, $P_1(\zeta)=0$ is solved by $\zeta=\frac{1}{x^*}$. The result follows by using Cauchy bound on the roots of a polynomial.
\end{proof}

\begin{lemma}\label{lemma:prodconvexity}
Let $\bx=(x_1,\dots,x_p)\sim F$, where $F$ is a multivariate density function. If $h_1,\dots,h_p>0$ be functions s.t. $\frac{\partial^2 \log(h_j(x_j))}{\partial x_j^2}>0$, then
$\prod_{j=1}^{p}h_j(x_j)$ is a convex as a multivariate function over $\bx$.
\end{lemma}
\noindent\emph{proof of lemma~\ref{lemma:prodconvexity}}
\begin{proof}
Note that $\frac{\partial \log(h_j(x_j))}{\partial x_j}=\frac{h_j'(x_j)}{h(x_j)}$ and $\frac{\partial^2 \log(h_j(x_j))}{\partial x_j^2}=\frac{h_j''(x_j)}{h_j(x_j)}-
\left(\frac{h_j'(x_j)}{h_j(x_j)}\right)^2$. This implies $\frac{h_j''(x_j)}{h_j(x_j)}=z_j^2+a_j$, $a_j=\frac{\partial^2 \log(h_j(x_j))}{\partial x_j^2}>0$ and
$z_j=\frac{h_j'(x_j)}{h_j(x_j)}$. Let $H(x_1,...,x_p)=\prod_{j=1}^{p}h_j(x_j)$. Then

{\footnotesize\begin{align*}
&\nabla^2 H(\bx)=\left(\begin{array}{cccc}
h_1''(x_1)\prod_{j\neq 1}h_j(x_j) & h_1'(x_1)h_1'(x_2)\prod_{j\neq 1,2}h_j(x_j) & \cdots & h_1'(x_1)h_p'(x_p)\prod_{j\neq 1,p}h_j(x_j)\\
h_1'(x_1)h_2'(x_2)\prod_{j\neq 1,2}h_j(x_j) & h_2''(x_2)\prod_{j\neq 2}h_j(x_j) & \cdots & h_2'(x_2)h_p'(x_p)\prod_{j\neq 2,p}h_j(x_j)\\
\vdots & \vdots & \vdots & \vdots\\
h_1'(x_1)h_p'(x_p)\prod_{j\neq 1,p}h_j(x_j) & h_p'(x_p)h_2'(x_2)\prod_{j\neq 2,p}h_j(x_j)& \cdots & h_p''(x_p)\prod_{j\neq p}h_j(x_j)\\
\end{array}\right)\\
&=\prod_{j}h_j(x_j)\left(\begin{array}{cccc}
\frac{h_1''(x_1)}{h_1(x_1)} & \frac{h_1'(x_1)h_2'(x_2)}{h_1(x_1)h_2(x_2)} & \cdots & \frac{h_1'(x_1)h_p'(x_p)}{h_1(x_1)h_p(x_p)}\\
\frac{h_1'(x_1)h_2'(x_2)}{h_1(x_1)h_2(x_2)} & \frac{h_2''(x_2)}{h_2(x_2)} & \cdots & \frac{h_2'(x_2)h_p'(x_p)}{h_2(x_2)h_p(x_p)}\\
\vdots & \vdots & \vdots & \vdots\\
\frac{h_1'(x_1)h_p'(x_p)}{h_1(x_1)h_p(x_p)} & \frac{h_p'(x_p)h_2'(x_2)}{h_2(x_2)h_p(x_p)}& \cdots & \frac{h_p''(x_p)}{h_p(x_p)}\\
\end{array}\right)\\
&=\prod_{j}h_j(x_j)\left(\begin{array}{cccc}
a_1+z_1^2 & z_1z_2 & \cdots & z_1z_p\\
z_1z_2 & a_2+z_2^2 & \cdots & z_2z_p\\
\vdots & \vdots & \vdots & \vdots\\
z_1z_p & z_2z_p & \cdots & a_p+z_p^2\\
\end{array}\right)\\
&=\prod_{j}h_j(x_j)\left\{\diag(a_1,\dots,a_p)+\left(\begin{array}{c}
z_1\\
\vdots\\
z_p
\end{array}\right)(z_1,\dots,z_p)\right\}.
\end{align*}}
Therefore $\nabla^2 H(\bx)$ is a positive definite matrix proving the lemma.
\end{proof}

\begin{lemma}\label{lemma:ModifBessel}
If $K_{\nu}(x)$ is the modified Bessel function of the second kind with parameter $\nu\in\Re$, then
\begin{align*}
2^{\nu-1}\Gamma(\nu)>x^{\nu}K_{\nu}(x)>2^{\nu-1}\Gamma(\nu) e^{-x},
\end{align*}
for all $x>0$ and $\nu>0$.
\end{lemma}
\begin{proof}
For a detailed proof, see Gaunt (2014).
\end{proof}

\begin{lemma}\label{lemma:normmeasure}
Suppose $T_1,\dots,T_m$ are independent random variables with $T_j$ having density $f_j$ supported in $(0,\infty)$. Let $\phi_j=\frac{T_j}{\sum_{l=1}^{m}T_m}$. Then the joint density of $(\phi_1,\dots,\phi_{m-1})$ has a joint density supported on the simplex $\mathcal{S}^{m-1}$ and is given by
\begin{align*}
f(\phi_1,\dots,\phi_{m-1})=\int_{t=0}^{\infty} t^{m-1}\prod_{l=1}^{m}f_j(\phi_j t) dt,
\end{align*}
where $\phi_m=1-\sum_{l=1}^{m-1}\phi_l$.
\end{lemma}
\noindent\emph{proof of lemma~\ref{lemma:normmeasure}}
\begin{proof}
This result is well known in the theory of normalized random measures \citep{kruijer2010adaptive, zhou2013negative}.
\end{proof} 

\end{document}

%% file: defs.tex
\newcommand{\KL}{\mathrm{KL}}
\renewcommand{\vec}{\mathrm{vec}}

\newtheorem{theorem}{Theorem}[section]
\newtheorem{lemma}[theorem]{Lemma}

\newenvironment{proof}[1][Proof]{\begin{trivlist}
\item[\hskip \labelsep {\bfseries #1}]}{\end{trivlist}}

\newenvironment{remark}[1][Remark]{\begin{trivlist}
\item[\hskip \labelsep {\bfseries #1}]}{\end{trivlist}}

\newcommand{\qed}{\nobreak \ifvmode \relax \else
      \ifdim\lastskip<1.5em \hskip-\lastskip
      \hskip1.5em plus0em minus0.5em \fi \nobreak
      \vrule height0.75em width0.5em depth0.25em\fi}

\newcommand{\Exp}{{\mathbb E}}

\newcommand{\balpha}{ {\boldsymbol \alpha} }

\newcommand{\bbeta}{ {\boldsymbol \beta} }

\newcommand{\bgamma}{ {\boldsymbol \gamma} }

\newcommand{\bphi}{ {\boldsymbol \phi} }

\newcommand{\blambda}{ {\boldsymbol \lambda} }
\newcommand{\bLambda}{ {\boldsymbol \Lambda} }
\newcommand{\bmu}{ {\boldsymbol \mu} }

\newcommand{\bomega}{ {\boldsymbol \omega} }

\newcommand{\bSigma}{ {\boldsymbol \Sigma} }
\newcommand{\btheta}{ {\boldsymbol \theta} }

\DeclareMathOperator{\diag}{diag}

\newcommand{\bzero}{ {\boldsymbol 0} }

\newcommand{\ba}{ {\boldsymbol a} }

\newcommand{\bb}{ {\boldsymbol b} }
\newcommand{\bB}{ {\boldsymbol B} }

\newcommand{\bF}{ {\boldsymbol F} }

\newcommand{\bH}{ {\boldsymbol H} }

\newcommand{\bO}{ {\boldsymbol O} }

\newcommand{\bT}{ {\boldsymbol T} }

\newcommand{\bW}{ {\boldsymbol W} }
\newcommand{\bx}{ {\boldsymbol x} }
\newcommand{\bX}{ {\boldsymbol X} }
\newcommand{\by}{ {\boldsymbol y} }

\newcommand{\bz}{ {\boldsymbol z} }
\newcommand{\bZ}{ {\boldsymbol Z} }

\newcommand{\norm}[1]{\left\lVert#1\right\rVert}

\newcommand{\var}{\mbox{var}}